\documentclass{article}
\usepackage[utf8]{inputenc}
\usepackage{float}
\usepackage{jheppub}
\usepackage{amsmath}
\usepackage[english]{babel}
\usepackage{amssymb}
\usepackage{mathtools}
\usepackage{color}
\definecolor{ao}{rgb}{0.0, 0.5, 0.0}

\usepackage{graphicx}
\usepackage{longtable}
\usepackage{tikz}
\usepackage{comment}
\usepackage{capt-of}
\usepackage{amsthm}
\usepackage{hyperref}
\hypersetup{
    colorlinks=true,
    linkcolor=blue
    }

\allowdisplaybreaks

\usepackage{bm}

\usepackage{caption}
\usepackage{subcaption}

\newtheorem{remark}{Remark}
\newtheorem{theorem}{Theorem}
\newtheorem{proposition}{Proposition}
\newtheorem{corollary}{Corollary}
\newtheorem{definition}{Definition}
\newtheorem{lemma}{Lemma}
\newtheorem{assumption}{Assumption}

\def\0{\mbox{\tiny $0$}}
\def\1{\mbox{\tiny $1$}}
\def\2{\mbox{\tiny $2$}}
\def\3{\mbox{\tiny $3$}}
\def\4{\mbox{\tiny $4$}}
\def\5{\mbox{\tiny $5$}}
\def\6{\mbox{\tiny $6$}}
\def\7{\mbox{\tiny $7$}}
\def\8{\mbox{\tiny $8$}}
\def\9{\mbox{\tiny $9$}}

\def\r{\rangle}

\def\l{\langle}
\def\m{\bar{m}}

\def\q{\bar{q}}

\newcommand{\SOMMA}[2]{\displaystyle\sum\limits_{#1}^{#2}}

\long\def \beq#1\eeq {\begin{equation} #1 \end{equation}}
\long\def \beaq#1\eeaq {\begin{equation}\begin{aligned} #1 \end{aligned}\end{equation}}
\long\def \bes#1\ees {\begin{equation}\begin{split} #1 \end{split} \end{equation}}
\long\def \bea#1\eea {\begin{eqnarray} #1 \end{eqnarray}}
\long\def \bse[#1]#2\ese {\begin{subequations}\label{#1}\begin{align} #2 \end{align}\end{subequations}}

\newcommand{\qb}{\bar{q}}
\newcommand{\mb}{\bar{m}}

\title{Guerra interpolation for  {\em inverse freezing}}
\author[1,2]{Linda Albanese,}
\author[3,4]{Adriano Barra,}
\author[3,4]{Emilio N. M. Cirillo}

\affiliation[1]{Dipartimento di Matematica e Fisica ``Ennio De Giorgi'', Universit\`a del Salento, Lecce, Italy.}
\affiliation[2]{Istituto Nazionale di Fisica Nucleare, Sezione di Lecce, Lecce, Italy.}
\affiliation[3]{Dipartimento di Scienze di Base e Applicate per l'Ingegneria, Sapienza Università di Roma, Italy.}
\affiliation[4]{Istituto Nazionale di Fisica Nucleare, Sezione di Roma1, Roma, Italy.}

\abstract{In these short notes we adapt and systematically apply Guerra's interpolation techniques on a class of disordered mean field spin glasses equipped with crystal fields and multi-value spin variables. These models undergo the phenomenon of {\em inverse melting} or  {\em inverse freezing}. In particular, we focus on the Ghatak-Sherrington model, its extension provided by Katayama and Horiguchi and the disordered Blume-Emery-Griffiths-Capel model in the mean-field limit deepened by Crisanti $\&$ Leuzzi and by Schupper $\&$ Shnerb. Once shown how all these models can be retrieved as particular limits of a unique broader Hamiltonian, we study their free energies. We provide explicit expressions of their annealed and quenched expectations, inspecting the cases of  replica symmetry and  (first step) broken-replica-symmetry. We recover several results previously obtained via heuristic approaches (mainly replica trick) to prove full agreement with the existing Literature. As a sideline, we inspect also the onset of replica symmetry breaking by providing analytically the expression of the De-Almeida and Thouless instability line for a replica symmetric description: in this general setting, the latter is new also from a physical viewpoint.}

\begin{document}

\maketitle

\section{Introduction}

{\em Inverse Freezing} -- or {\em Inverse Melting} (depending if we deal with a system, respectively, with or without quenched disorder) --  is a rare and counterintuitive phenomenon where a liquid crystallizes (or vitrifies, i.e., {\em glassifies}) upon heating \cite{SchupperShnerb}. In order to obey Thermodynamics, this phenomenon must occur in systems whose crystalline or glassy phase possesses  higher entropy than their liquid phase. A possible mechanism accounting for this reverse transition is the freezing of some molecular degrees of freedom in the liquid state, which become active in the solid phase, leading to an overall increase in entropy upon crystallization. This concept extends to {\em inverse glass transitions}, where a system enters a glassy state as temperature increases (instead of decreasing) when quenched disorder --typical of  spin glasses \cite{MPV}-- is also present.  
\newline
As for inverse melting, the increased  disorder of atomic positions in the liquid is compensated by  greater order in another characteristics, such as center-of-mass localization. For example, if a crystalline phase has a larger molar volume  than the liquid, its vibrational entropy (associated with thermal oscillations) may be  higher than the configurational entropy loss, leading to crystallization upon heating. 
\newline
Although uncommon,  inverse freezing has  been empirically observed in disparate systems, ranging from e.g. sticky hard spheres \cite{hard1,hard2} to proteins undergoing cold denaturation \cite{TCS1,TCS2}.
\newline
Beyond experimental evidence, up to now the bulk of theoretical research on this weird thermodynamical phenomenon is built of mainly by heuristic approaches as e.g. the replica trick \cite{MPV}  and almost no rigorous results are available in the Literature\footnote{Notably, in \cite{Travello} the Authors study inverse freezing in a quantum fermionic van Hemmen spin glass without relying upon the replica trick (they use Grassmann's fields \cite{Irene}).}.
\newline
Aim of this paper is to contribute to fill this gap by adapting Guerra interpolation procedures (see e.g. \cite{barra2010replica,guerra_broken,guerra1996overlap,B-war1,B-war2,B-war3,B-war4})  to deal with a large family of already existing models accounting for these inverse transitions, in order to confirm replica outcomes. Specifically, the models we address in this paper are the disordered  mean-field Blume-Capel model as treated by Schupper and Shnerb \cite{schupper2004spin} (see also \cite{SchupperShnerb} and \cite{EmilioEnzo1,EmilioNardi2}, or \cite{Kaufman1990} for a version of the Blume-Capel model in a random field), the Ghatak-Sherrington model \cite{ghatak1977crystal} (see also \cite{sheng2024fluctuation}) and its generalization provided by the Katayama-Horiguchi model \cite{katayama1999ghatak}  as well as the Blume-Emery-Griffith-Capel model in the disordered mean field limit analyzed by Crisanti and Leuzzi \cite{crisanti-leuzziFullRSB_PRB}  (see also \cite{crisanti2005stable,leuzzi2007spin,Paoluzzi2004,Branco}). 
\newline
Furthermore, rather than inspecting all these models one by one, we introduce a generalized Hamiltonian that recovers all of them as particular limits of a broader picture such that we can keep just one model (the present generalization) and one technique (the present interpolation) to paint the equilibrium properties of these systems by a statistical mechanical perspective.
\newline
Finally, as we aim to obtain an independent proof of a physical picture already existing and discussed elsewhere, we remind to the original papers for a discussion about the underlying physical insights while, in the following, we focus solely on the mathematical methodological aspects that lie behind the physical outcomes: as standard in equilibrium statistical mechanics of disordered models \cite{MPV}, the main observable we deal with is the free energy of the model as its extremization w.r.t. the order parameters of the theory allows to paint phase diagrams in the space of the model's control parameters, ultimately exhaustively describing its thermodynamic behavior at equilibrium. Evaluating the free energy is, in general, of prohibitive difficulty for disordered systems even at the mean field level \cite{guerra_broken,talagrand2000rigorous} and achieving a  collection of approximations is the usual path to comprehend these systems. Starting from the annealed expression, we then move to evaluate the quenched one, both under the replica symmetry  as well as under the broken replica symmetry \textit{Ansatz} (confined to the first step of replica symmetry breaking (1-RSB) only). As a technical detail, as replica symmetry breaking (RSB) is expected to play a non marginal role in this context, we also derive the De Almeida and Thouless (AT) instability line of the replica symmetric (RS) solution\footnote{We stress, as a rather technical remark, that --in order to derive the instability line for the RS scenario-- we do not rely nor upon the standard argument provided by De Almeida and Thouless \cite{de1978stability}, neither by the Kondor's perspective \cite{Kondor-arxiv2022}, rather we follow the path we recently developed in \cite{albanese2023almeida}, intrinsically closer in spirit to the Toninelli \cite{toninelli2002almeida} or Chen \cite{chen2021almeida} arguments.} so to split, in the space of the control parameters of the model (i.e. in the phase diagram), a region where such a symmetry is preserved (and, thus, the RS description of the quenched free energy is useful) from a region where this symmetry is spontaneously broken (and, thus, the 1-RSB expression for the quenched free energy must be preferred over the RS one). 
\par\medskip
The paper is structured as follows: in Section \ref{sec:generalities} we introduce the unified model that we address in the paper as well as its related statistical mechanical package of definitions. Section \ref{sec:annealed} is entirely dedicated to the annealed expression of the free energy: while this approximation holds solely in the high noise limit (hence its knowledge provides marginal comprehension), yet it can be calculated exactly and --due to Jensen inequality-- it also plays as an upper bound for the quenched one. Then, Section \ref{sec:Guerra} is completely devoted to the inspection of the quenched free energy: in particular, in Sec. \ref{ssec:RS} we provide its RS expression, while in Sec. \ref{sec:1-RSB}  we generalize it under the 1-RSB prescription and in \ref{sec:ATline} we derive the instability region for the RS picture, namely the boundary between the two descriptions. 
\newline
Conclusions and outlooks are discussed in Section \ref{Conclusioni} that closes the paper: in order to streamline the main body of the manuscript, we wrote also an Appendix where lengthy calculations are reported.

\section{A unified model for inverse freezing}
\label{sec:generalities}
Let us start by introducing the model and its related package of statistical mechanical definitions.
\begin{definition}
\label{def:model}
 Using the label $s$ for the dynamical variables (i.e. the spins), the label $J$ for the quenched variables (i.e. the couplings) and the label $h$ for external inputs (e.g. a magnetic field), the Hamiltonian that accounts for the \textit{inverse freezing} phenomenon (or \textit{inverse melting} if no disorder is present) is defined as 
    \begin{align}
        \mathcal H_{N, D, K}(\bm s \vert \bm J, \bm h) := - \dfrac{1}{2}\sum_{i,j=1}^N J_{ij} s_i s_j + D \sum_{i=1}^N s_i^2 - \sum_{i=1}^N h_i s_i - \dfrac{K}{N} \sum_{i<j} s_i^2 s_j^2,
        \label{eq:Hamiltonian}
    \end{align}
    where $s_i \in \Omega=\{-1 + \dfrac{\gamma}{S}, \ \gamma=0, \hdots, 2S\}$, with $S \geq \dfrac{1}{2}$, $J_{ij} \sim \mathcal{N}\left(\dfrac{J_0}{N}, \dfrac{J^2}{N}\right)$ i.i.d., $i\neq j =1, \hdots, N$, $J_{ii}=0$, for each $i=1, \hdots, N$, $D, \ K \in \mathbb{R}$ and $h_i$ are i.i.d. random variables whose probability distribution $\mathbb{P}(h)$ is assumed to be known\footnote{Note that we inspect also the standard case of deterministic external field, i.e.  $\mathbb{P}(h) = \delta (h-H)$, for some value $H \in \mathbb{R}$.}.
\end{definition}

\begin{remark}
    Let us analyse the role of $S$. We need to distinguish two cases. 
    \begin{itemize}
        \item $S \in \mathbb{N}$. In this case the value $0$ is always included in the set $\Omega$ and $|\Omega|$ is always odd.
        \item $S=n+\dfrac{1}{2}, \ n \in \mathbb{N}_0$. In this case $0$ is not included in the possible values of the set $\Omega$ and $|\Omega|$ is even. 
    \end{itemize}
\end{remark}

\begin{remark} We list hereafter the various models that Hamiltonian \eqref{eq:Hamiltonian} reduces to in the suitable limits.
\begin{itemize}
\item  For $S=1/2$, $h_i=H \in \mathbb{R}\ \forall i \in (1,...,N)$ and $D=K=J_0=0$, the Hamiltonian \eqref{eq:Hamiltonian} reproduces the Sherrington-Kirkpatrick (SK) model\footnote{We stress that, actually, the SK model is not a candidate model to undergo inverse freezing. Nevertheless, as it plays as the  {\em harmonic oscillator} for spin glasses both from a physical \cite{MPV} as well as mathematical \cite{talagrand2003spin} point of view, we keep the study of this limit as a reasonable checkpoint during our investigation.} \cite{sherrington1975solvable}).     

\item  For $K=0$, $J_0=0$, $S = 1$ and $h_i=0, \forall i=1, \hdots , N$, the Hamiltonian \eqref{eq:Hamiltonian} reproduces the GS model \cite{ghatak1977crystal} (see  Eq. $(1)$ there) and the model studied by Schupper and Shnerb \cite{schupper2004spin}, (see their Eq. $(4)$). 

\item For $K=0$, $S = 1$ and $h_i$ as i.i.d. Bernoullian random variables $\forall i=1, \hdots , N$, the Hamiltonian  \eqref{eq:Hamiltonian} reproduces the generalization of the Ghatak-Sherrington (GS) model studied by Morais and coworkers in \cite{morais2013inverse} (see Eq. (1) in that paper). Moreover, if we also suppose $J_{ij}=\dfrac{1}{N}$, $i,j=1, \hdots, N$, we recover the Hamiltonian studied by Kaufman and Kanner in \cite{Kaufman1990}.

\item For $K=0$, $S = 1$ and $h_i$ i.i.d. Gaussian random variables $\forall i=1, \hdots , N$, the Hamiltonian  \eqref{eq:Hamiltonian} reproduces the generalization of the GS model provided by Morais and coworkers in \cite{morais2012inverse} (see Eq. (1) in that paper). 

\item For $K=0$ and $h_i=H \in \mathbb{R}$, the Hamiltonian  \eqref{eq:Hamiltonian} reproduces the GS model with $S$ spins deepened by Katayama and Horiguchi \cite{katayama1999ghatak}.

\item Finally, if we consider $h_i = 0$, for all $i=1, \hdots, N$, $J_0=0$, $J^2=1$ and $S=1$, we get the mean-field Blume-Every-Griffiths-Capel model (see Eq. (1) in \cite{leuzzi2007spin} or Eq. (1) in \cite{CrisantiLeuzzi05}).


\end{itemize}
\end{remark}


In order to quantitatively describe the global behavior of the model, we need to introduce \textit{order parameters} and \textit{control parameters} as follows.
\begin{definition}
\label{def:orderparam}
The order parameters of the  Hamiltonian model \eqref{eq:Hamiltonian} are the magnetization $m$ and the two replica overlap $q^{a,b}$, defined respectively as 
    \begin{align}\label{LaMagnetizzazione}
        m_N(\bm s)&:=\dfrac{1}{N} \sum_{i=1}^N s_i,\ \ \ \  \ \ \ m(\bm s)=\lim_{N \to \infty} m_N(\bm s)\\ \label{Loverlappo}
        q_N^{a,b}(\bm s)&:= \dfrac{1}{N} \sum_{i=1}^N s_i^{(a)} s_i^{(b)},  \ \ \ q^{a,b}=\lim_{N \to \infty} q_N^{a,b}(\bm s),
    \end{align}
while the control parameters are $J_0,\ J^2,\ h,\  D,\  K$ and the inverse temperature $\beta$ ({\em vide infra}).
\end{definition}
We note that we use the subscript $N$ for a generic observable $O_N(\bm s|\bm J, \bm h)$ when evaluated at finite volume, while $O(\bm s|\bm J, \bm h)$ depicts its asymptotic behavior as $N \to \infty$.
\newline
The magnetization $m_N$ is the empirical average of the values of all the spins. Instead, the two-replicas overlap $q_N^{a,b}$ quantifies how close two replicas\footnote{Two replicas are two copies of the system sharing the same realization of the quenched variables (i.e. $\bm J$ and $\bm h$).} are.   A particular case of two-replicas overlap is when $\bm s^{(a)}$ and $\bm s^{(b)}$ are the same, namely 
\begin{align}
    q_N^{a,a}(\bm s):= \dfrac{1}{N} \sum_{i=1}^N \left(s_i^{(a)}\right)^2, \ \ \ \ q^{a,a}(\bm s) =  \lim_{N \to \infty} q_N^{a,a}(\bm s),
\end{align}
which keeps track of the total number of non-zero entries.
\newline
From now on, we omit the  explicit dependence of the order parameters by $\bm s$ to lighten the notation.

\bigskip

{\em En route} to define the free energy, hereafter we introduce the Boltzmann average and its related concept of partition function. 
\begin{definition}
We denote the Boltzmann factor associated to the Hamiltonian \eqref{eq:Hamiltonian}, at inverse temperature\footnote{For historical reason  we use $\beta$ that is the inverse of the temperature $T$.} $\beta:=1/T \in \mathbb{R}^{+}$, as
\begin{equation}
{B}_{N, D, K}(\bm{s}, \beta |\bm{J}, \bm h):={e^{-\beta \mathcal{H}_{N, D, K}(\bm s \vert \bm J , \bm h)}},\quad\quad\quad Z_{N, D, K}(\beta \vert \bm J, \bm h)=\sum_{\{\bm s\}}e^{-\beta \mathcal{H}_{N, D, K}(\bm s \vert \bm J, \bm h )},
\label{eq:defBZ}
\end{equation}
where the normalization factor $Z_{N, D, K}(\beta \vert \bm J, \bm h)$ is called partition function. 

The corresponding Boltzmann-Gibbs average and quenched average, for a generic observable $O_N(\bm s|\bm J, \bm h)$, are defined as 
\begin{align}
    \omega^{(\beta \vert \bm J, \bm h)}_{N, D, K}\left(O_N(\bm s|\bm J, \bm h)\right):=&\dfrac{1}{Z_{N, D, K}(\beta \vert \bm J, \bm h)} \sum_{\{\bm s\}} {B}_{N, D, K}(\beta, \bm{s}|\bm{J}, \bm h) O_N(\bm s|\bm J, \bm h) \label{eq:BG} \\
    \langle O_N(\bm s|\bm J, \bm h) \rangle := &\mathbb{E}_{J, \bm h} \left[\omega^{(\beta \vert \bm J, \bm h)}_{N, D, K}\left(O_N(\bm s|\bm J, \bm h)\right)\right] \label{eq:quenched}
\end{align}
where $\mathbb{E}_{J, \bm h}$ is the expectation with respect to all i.i.d. $J_{ij}$ and $h_i$ random variables. 
\end{definition}
We are now ready to introduce the main observables of the manuscript. From now on, we omit the explicit dependence on $\beta, \ {N,\  D,\  K,\  \bm J}$ and $\bm h$ of the Boltzmann average to lighten the notation. 
\begin{definition}
Once the partition function is introduced, we can define the free energy $f_{D,K}(\beta)$ or, equivalently, the  statistical pressure $\mathcal{A}_{D,K}(\beta)$, as 
\begin{align}
    \mathcal{A}_{D, K}(\beta)&:= \lim_{N \to +\infty}\mathcal{A}_{N, D, K}(\beta) = \lim_{N \to +\infty} \left[\dfrac{1}{N} \log Z_{N, D, K}(\beta \vert \bm J, \bm h) \right]\notag \\
    &=\lim_{N \to +\infty}\left[-\beta f_{N, D, K}(\beta)\right]=: -\beta f_{D, K}(\beta),
    \label{eq:A}
\end{align}
where $f_{N, D, K}$ is known as the free energy at size $N$.  
\newline
Their quenched expectations, $\mathcal{A}^Q_{D,K}(\beta)$ and $f^Q_{D, K}(\beta)$ read as
\begin{align}
    \mathcal{A}^Q_{D,K}(\beta)&:= \lim_{N \to +\infty}\mathcal{A}^Q_{N, D, K}(\beta) = \lim_{N \to +\infty} \left[\dfrac{1}{N} \mathbb{E}_{\bm J, \bm h} \log Z_{N, D, K}(\beta \vert \bm J, \bm h) \right]\notag \\
    &=\lim_{N \to +\infty}\left[-\beta f_{N, D, K}^Q(\beta)\right]=:-\beta f^Q_{D,K}(\beta),
    \label{eq:Adef}
\end{align}
where $f^Q_{N, D, K}$ is known as the quenched free energy at size $N$. 
\end{definition}
\begin{remark}
As the free energy (or, equivalently, the statistical pressure) is a self-averaging quantity, see e.g, \cite{Cavagna}, in the thermodynamic limit the physical properties of these disordered models are expected not to depend on the particular  realization of the quenched couplings $\bm J$, i.e. by taking advantage of the Chebyshev's inequality  
$$
\mathbb{P}\left( \left|\frac{1}{\beta N} \ln  Z_{N, D, K}(\beta \vert \bm J, \bm h) - \frac{1}{\beta N} \mathbb{E} \ln  Z_{N, D, K}(\beta \vert \bm J, \bm h)   \geq x   \right| \right) \leq e^{-N\frac{x^2}{2}},
$$
hence we focus on the quenched expectation of the free energy rather than the free energy itself that is a random quantity\footnote{We stress that the above argument, at present, is just a heuristic reasoning and not a rigorous statement. This is because, for the model \eqref{eq:Hamiltonian}, nor the full Parisi solution exists (in \cite{leuzzi2007spin,CrisantiLeuzzi04} Crisanti and Leuzzi worked out the full-RSB free energy of the disordered Blume–Emery–Griffiths model in the mean field  limit, but with the usage of the replica trick) neither a proof of the existence of the infinite volume limit of the quenched free energy is available (i.e. the Guerra-Toninelli approach \cite{guerra2002thermodynamic} does not apply straightforwardly in models with diluted couplings or diluted spins), hence --by now-- we must assume the existence of the quenched expectation.}.     
\end{remark}
Purpose of this manuscript is thus to obtain explicit expressions of the quenched free energy, in the thermodynamic limit, in terms of its control and order parameters as this knowledge suffices to draw the phase diagram of the model. In our journey, we first derive its annealed approximation in Sec. \ref{sec:annealed}, then we move to the quenched expectation, that is tackled in Sec. \ref{sec:Guerra}. 
Finally, we point out that the study of the annealed approximation is useful as the latter is a bound for the quenched free energy, as explained in Remark \ref{JensenRemark} in the next Section.

\begin{remark}
    For the sake of simplifying computations, it turns useful to rewrite the Hamiltonian of the model \eqref{eq:Hamiltonian}, switching from Normal to Standard random variables, as 
    \begin{align}
        \mathcal H_{N, D, K, J, J_0}(\bm s \vert \bm z, \bm h) = - \dfrac{J}{\sqrt{2N}}\sum_{i ,j=1}^N z_{ij} s_i s_j  - \dfrac{J_0}{2N} \sum_{i,j = 1}^N s_i s_j+ D \sum_{i=1}^N s_i^2 - \sum_{i=1}^N h_i s_i - \dfrac{K}{N} \sum_{i<j=1}^N s_i^2 s_j^2,
    \end{align}
    with $z_{ij}$ i.i.d. Standard variables\footnote{As already done in the Literature, we have supposed now that the variance of $J_{ij}$ is $\frac{2J^2}{N}$.}.  
\end{remark}

\section{The annealed free energy}
\label{sec:annealed}

In this Section, at first, rather than focusing directly on the quenched statistical pressure (or free energy)  we provide the analytical expression of the annealed statistical pressure (or free energy) related to the model \eqref{eq:Hamiltonian} as the latter plays as a bound for the former as stated by the following remark. 
\newline
We stress that the annealed approximation of the free energy for this broad class of candidate models to undergo inverse freezing constitute, to our knowledge, a novelty in the Literature. 
\begin{definition}
The annealed statistical pressure for the Hamiltonian model \eqref{eq:Hamiltonian} is defined as
    \begin{align}
        \mathcal{A}_{N, D, K, J, J_0}^A(\beta):= - \beta f^A_{N, D, K, J, J_0}(\beta)  := \dfrac{1}{N} \mathbb{E}_{\bm h} \log \mathbb{E}_{\bm z} Z_{N, D, K, J, J_0}(\beta \vert \bm z, \bm h)
        \label{eq:annealed}
    \end{align}
    where $Z_{N, D, K, J, J_0}(\beta \vert \bm z, \bm h)$ is the corresponding partition function \eqref{eq:defBZ} and $f^A_{N, D, K, J, J_0}(\beta)$ the corresponding annealed free energy. 
\end{definition}
\begin{remark}\label{JensenRemark}
Since the function $x \to \ln(x)$ is concave, due to Jensen inequality, 
\begin{align*}
f_{N, D, K, J, J_0}^A (\beta)= -\frac{1}{\beta N}\mathbb{E}_{\bm h}\log &\mathbb{E}_{\bm z} Z_{N, D, K, J, J_0}(\beta \vert \bm z, \bm h) \notag \\
&\leq -\frac{1}{\beta N} \mathbb{E}_{\bm h, \bm z}\log   Z_{N, D, K, J, J_0}(\beta \vert \bm J, \bm h)=f_{N, D, K, J, J_0}^Q(\beta), 
\end{align*}
where   $f_{N, D, K, J, J_0}^Q(\beta)$ is the quenched free energy.
\newline
Note that, as $\mathcal{A}_{ D, K, J, J_0}(\beta) = - \beta f_{D, K, J, J_0}(\beta)$, the bound is reversed for  the statistical pressure $\mathcal{A}_{D, K, J, J_0}(\beta)$.
\end{remark}
\begin{proposition}
\label{prop:annealed}
    The annealed statistical pressure of the Hamiltonian model \eqref{eq:Hamiltonian} at finite size $N$ for $S \geq 1/2$ is 
        \begin{align}
        \mathcal{A}^A_{N, D, K, J, J_0}(\beta) =&\dfrac{1}{N}\mathbb{E}_h\log \mathbb{E}_x \mathbb{E}_y \mathbb{E}_w \left\{\sum_{s \in \Omega} \exp \left[s^2 \left( -\beta D + \dfrac{\beta J}{\sqrt{2N}} w + y\sqrt{\dfrac{\beta K}{
        N}}\right)\right.\right. \notag \\
        &\hspace{2.5cm}\left.\left.+s \left(\beta h  + \sqrt{\dfrac{\beta J_0 }{N}}x \right)\right]\right\}^N,
        \label{eq:annealed_final}
    \end{align}
    where $x, \ y, \ w$ are standard Gaussian variables. The infinite volume limit of the annealed statistical pressure reads
    \begin{align}
    \mathcal{A}^A_{D, K, J, J_0}(\beta) =& \lim_{N \to +\infty} \mathcal{A}^A_{N,D, K, J, J_0}(\beta)=\mathbb{E}_h\log\left[\SOMMA{s \in \Omega}{} \exp\left({\beta h s-\beta D s^2}\right)\right] \notag
    \\
    &+\dfrac{\beta J_0}{2}\mathbb{E}_h\left[ \dfrac{\SOMMA{s \in \Omega}{} \exp\left({\beta h s -\beta D s^2}\right)s}{\SOMMA{s \in \Omega}{} \exp\left({\beta h s -\beta D s^2}\right)}\right]^2\notag \\
    &+\left(\dfrac{\beta^2 J^2}{4}+\dfrac{\beta K}{2}\right)\mathbb{E}_h\left[   \dfrac{\SOMMA{s \in \Omega}{} \exp\left({\beta h s -\beta D s^2}\right)s^2}{\SOMMA{s \in \Omega}{}\exp \left({\beta h s-\beta D s^2}\right)}\right]^2.
    \label{eq:annealed_f}
\end{align}
\end{proposition}

\begin{remark}
    The annealed statistical pressure \eqref{eq:annealed_f} can be rewritten slightly differently depending on the value assumed by $S$. Indeed, for $S \in \mathbb{N}$, \eqref{eq:annealed_f} can be rewritten as 
\begin{align}
    \mathcal{A}^A_{N, D, K, J, J_0}(\beta) =&\mathbb{E}_h\log\left(1+2\SOMMA{\gamma=0}{S-1} \cosh[\beta h (-1 + \gamma/S) ]e^{-\beta D (1 - \gamma/S)^2}\right)  \notag 
    \\
    &+\dfrac{\beta J_0}{2}\mathbb{E}_h\left[ \dfrac{2\SOMMA{\gamma=0}{S-1} \sinh[\beta h (-1 + \gamma/S) ]e^{-\beta D (1 - \gamma/S)^2}\left( -1+ \dfrac{\gamma}{S}\right)}{1+2\SOMMA{\gamma=0}{S-1}\cosh[\beta h (-1 + \gamma/S) ] e^{-\beta D (1 - \gamma/S)^2}}\right]^2\notag
    \\
    &+\left(\dfrac{\beta^2 J^2}{4}+\dfrac{\beta K}{2}\right)\mathbb{E}_h \left[  \dfrac{2\SOMMA{\gamma=0}{S-1}\cosh[\beta h (-1 + \gamma/S) ]e^{-\beta D (1 - \gamma/S)^2}\left( -1+ \dfrac{\gamma}{S}\right)^2}{1+2\SOMMA{\gamma=0}{S-1}\cosh[\beta h (-1 + \gamma/S) ]e^{-\beta D (1 - \gamma/S)^2}}\right]^2.
    \end{align}
    Instead, for $S \in \mathbb{N} + \dfrac{1}{2}$ we have 
    \begin{align}
    \mathcal{A}^A_{N, D, K, J, J_0}(\beta) =&\log 2+\mathbb{E}_h \log\left(\SOMMA{\gamma=0}{S-1/2}\cosh[\beta h (-1 + \gamma/S) ]e^{-\beta D (1 - \gamma/S)^2}\right)  \notag 
    \\
    &+\dfrac{\beta J_0}{2}\mathbb{E}_h \left[ \dfrac{\SOMMA{\gamma=0}{S-1/2}\sinh[\beta h (-1 + \gamma/S) ]e^{-\beta D (1 - \gamma/S)^2}\left( -1+ \dfrac{\gamma}{S}\right)}{\SOMMA{\gamma=0}{S-1/2}\cosh[\beta h (-1 + \gamma/S) ] e^{-\beta D (1 - \gamma/S)^2}}\right]^2\notag
    \\
    &+\left(\dfrac{\beta^2 J^2}{4}+\dfrac{\beta K}{2}\right)\mathbb{E}_h \left[   \dfrac{\SOMMA{\gamma=0}{S-1/2}\cosh[\beta h (-1 + \gamma/S) ]e^{-\beta D (1 - \gamma/S)^2}\left( -1+ \dfrac{\gamma}{S}\right)^2}{\SOMMA{\gamma=0}{S-1/2}\cosh[\beta h (-1 + \gamma/S) ]e^{-\beta D (1 - \gamma/S)^2}}\right]^2.
    \end{align}
\end{remark}



Now, let us thoroughly analyse Proposition \ref{prop:annealed} as we saw how, by selecting suitable values of the control parameters in the Hamiltonian \eqref{eq:Hamiltonian}, we recover several models of the inverse freezing Literature. For the sake of clearness, at first we always check that the SK limit of the theory we are developing returns the correct expression: despite the SK model does not undergo inverse freezing, it plays as the harmonic oscillator for these complex systems hence inspecting its correctness constitutes a reasonable preliminary check.
\begin{corollary}\label{Prop:SK-annealed}
    For $S=1/2$, $h_i=H \in \mathbb{R}$ and $D=K=J_0=0$ in \eqref{eq:Hamiltonian}, we recover the expression of the Hamiltonian of SK model \cite{sherrington1975solvable}
    \begin{align}
        \mathcal H_{N, J, H}(\bm s \vert \bm z) = - \dfrac{J}{2\sqrt{N}}\sum_{i j=1}^N z_{ij} s_i s_j 
        - H \sum_{i=1}^N  s_i,
        \label{eq:Hamiltonian_SK}
    \end{align}
    whose annealed statistical pressure reads 
    \begin{align}
        \mathcal{A}_{N, J, H}^A(\beta) = 
        \log 2 + \dfrac{\beta^2 J^2 }{4}+\log \left[\cosh (\beta H)\right],
 \label{eq:annealed_SK}
    \end{align}
    and in the thermodynamic limit simply $\mathcal{A}_{N, J, H}^A(\beta) \to \mathcal{A}^A_{J, H}(\beta)$, but the explicit expression does not change.
\end{corollary}

\begin{corollary}
\label{cor:GSannealed}
 Considering $K=0$ and $S = 1$ (namely $s_i \in \{-1, 0, +1\}$) in the Hamiltonian model \eqref{eq:Hamiltonian}, we obtain the GS model \cite{ghatak1977crystal, schupper2004spin} with a random field\footnote{The mean of $J_{ij}$ in \cite{ghatak1977crystal} is considered to be null, here the zero value is just a particular case.}: 
     \begin{align}
        \mathcal H_{N, D}(\bm s \vert \bm J, \bm h) := - \sum_{i\neq j=1}^N J_{ij} s_i s_j + D \sum_{i=1}^N s_i^2 - \sum_{i=1}^N h_i s_i,
        \label{eq:HamiltonianGS}
    \end{align}
    whose annealed statistical pressure, in the  thermodynamic limit $N \to \infty$, reads
    \begin{align}
    \mathcal{A}^A_{D, K, J, J_0}(\beta )=& \mathbb{E}_h\log \left[ 1+ 2 e^{-\beta D}  \cosh(\beta h)\right] +   \dfrac{\beta^2 J^2}{4} \mathbb{E}_h\left(\dfrac{2e^{-\beta D} \cosh(\beta h)}{1+2 \cosh(\beta h) e^{-\beta  D }}\right)^2 \notag \\
    &+ \dfrac{\beta J_0 }{2} \mathbb{E}_h\left(\dfrac{  2e^{-\beta D}\sinh(\beta h)}{1+2\cosh(\beta  h) e^{-\beta  D}} \right)^2 
        \label{eq:annealedGS2}      
    \end{align}
\end{corollary}
Note that the above expression \eqref{eq:annealedGS2} was lacking in the Literature.

Moreover, let us fix in Cor. \ref{cor:GSannealed} the probability distribution $\mathbb{P}(h)$: two notable examples are the following 
\begin{corollary}
\label{cor:annealed_bimodal}
     Let us suppose the probability distribution of $h$ to be  Bernoulli-like, e.g.
    \begin{align}
        \mathbb{P}(h)=p \delta(h-h_0)+(1-p) \delta(h+h_0),
    \end{align}
    with the mean $h_0 \in \mathbb{R}$ and $p \in [0,1]$.
    \newline
    The explicit expression of the  annealed statistical pressure \eqref{eq:annealedGS2} in the thermodynamic limit $N\to \infty$ becomes
    \begin{align}
         \mathcal{A}^A_{D, K, J, J_0, h_0}(\beta) =&\log \left[ 1+ 2 e^{-\beta D} \cosh(\beta h_0)\right] +\dfrac{\beta^2 J^2}{4} \left(\dfrac{2e^{-\beta D} \cosh(\beta h_0)}{1+2 \cosh(\beta h_0) e^{-\beta  D}}\right)^2 \notag \\
         &+ \dfrac{\beta J_0 }{2} \left(\dfrac{  2e^{-\beta D} \sinh(\beta h_0)}{1+2 \cosh(\beta  h_0) e^{-\beta  D}} \right)^2.
    \end{align}


We note that the Authors in \cite{morais2013inverse}\footnote{The Authors focused directly on the RS expression of the quenched free energy, a result that we (re)-obtain in the next Section via Guerra interpolation, see Sec. \ref{sec:Guerra}} did not focus explicitly on the annealed approximation of the free energy in their work and that the above model --with a Bernoulli random field-- has been introduced in \cite{Kaufman1990} (with the choice $J_{ij} \sim \mathcal{N}(0,\frac{1}{N})$).
\end{corollary}

\begin{proof}
In order to prove Cor. \ref{cor:annealed_bimodal}, we need simply to compute the averages for each term in \eqref{eq:annealedGS2}:
    \begin{align}
    \mathcal{A}^A_{D, K, J, J_0, h_0}(\beta)=&p\log \left[ 1+ 2 e^{-\beta D}  \cosh(\beta h_0)\right] +(1-p)\log \left[ 1+ 2 e^{-\beta D}  \cosh(-\beta h_0)\right]  \notag \\
    &+p\dfrac{\beta^2 J^2}{4} \left(\dfrac{2e^{-\beta D} \cosh(\beta h_0)}{1+2 \cosh(\beta h_0) e^{-\beta  D }}\right)^2 +(1-p)\dfrac{\beta^2 J^2}{4} \left(\dfrac{2e^{-\beta D} \cosh(-\beta h_0)}{1+2 \cosh(-\beta h_0) e^{-\beta  D }}\right)^2 \notag \\
    &+ \dfrac{\beta J_0 }{2} p\left(\dfrac{  2e^{-\beta D}\sinh(\beta h_0)}{1+2\cosh(\beta  h_0) e^{-\beta  D}} \right)^2+ \dfrac{\beta J_0 }{2} (1-p)\left(\dfrac{  2e^{-\beta D}\sinh(-\beta h_0)}{1+2\cosh(-\beta  h_0) e^{-\beta  D}} \right)^2  \notag \\
    =&p\log \left[ 1+ 2 e^{-\beta D}  \cosh(\beta h_0)\right] +(1-p)\log \left[ 1+ 2 e^{-\beta D}  \cosh(\beta h_0)\right]  \notag \\
    &+p\dfrac{\beta^2 J^2}{4} \left(\dfrac{2e^{-\beta D} \cosh(\beta h_0)}{1+2 \cosh(\beta h_0) e^{-\beta  D }}\right)^2 +(1-p)\dfrac{\beta^2 J^2}{4} \left(\dfrac{2e^{-\beta D} \cosh(\beta h_0)}{1+2 \cosh(\beta h_0) e^{-\beta  D }}\right)^2 \notag \\
    &+ \dfrac{\beta J_0 }{2} p\left(\dfrac{  2e^{-\beta D}\sinh(\beta h_0)}{1+2\cosh(\beta  h_0) e^{-\beta  D}} \right)^2+ \dfrac{\beta J_0 }{2} (1-p)\left(\dfrac{  -2e^{-\beta D}\sinh(\beta h_0)}{1+2\cosh(\beta  h_0) e^{-\beta  D}} \right)^2,
    \end{align}
    where in the last passage we have exploited the evenness and the oddness, respectively, of hyperbolic cosine and hyperbolic sine. With some trivial algebra we get 
    \begin{align}
    \mathcal{A}^A_{D, K, J, J_0, h_0}(\beta)=&\log \left[ 1+ 2 e^{-\beta D}  \cosh(\beta h_0)\right] +\dfrac{\beta^2 J^2}{4} \left(\dfrac{2e^{-\beta D} \cosh(\beta h_0)}{1+2 \cosh(\beta h_0) e^{-\beta  D }}\right)^2  \notag \\
    &+ \dfrac{\beta J_0 }{2} \left(\dfrac{  2e^{-\beta D}\sinh(\beta h_0)}{1+2\cosh(\beta  h_0) e^{-\beta  D}} \right)^2,
\end{align}
which is the thesis. 
\end{proof}

\begin{corollary}
\label{cor:annealed_Gauss}
     Supposing the probability distribution of $h$ to be Gauss-like, i.e.,
     \begin{align}
        \mathbb{P}(h) = \sqrt{\dfrac{1}{2\pi \Delta^2}} \exp \left( -\dfrac{1}{2\Delta^2} (h-h_0)^2\right),
    \end{align}
    with mean $h_0 \in \mathbb{R}$ and standard deviation $\Delta \in \mathbb{R}^+$.
    \newline
    The expression of the annealed statistical pressure \eqref{eq:annealedGS2} in the thermodynamic limit  $N \to \infty$  becomes
    \begin{align}
    \mathcal{A}^A_{ D, K, J, J_0, h_0, \Delta}(\beta)=&\mathbb{E}_h\log \left[ 1+ 2 \exp\left(-\beta D  \right) \cosh(\beta h_0 + \beta h \Delta)\right]  \notag \\
    &+   \dfrac{\beta^2 J^2}{4} \mathbb{E}_h\left[\dfrac{2\exp\left(-\beta D \right) \cosh(\beta h_0 + \beta h \Delta )}{1+2 \exp\left(-\beta D \right) \cosh(\beta h_0 + \beta h \Delta )}\right]^2\notag \\
    &+ \dfrac{\beta J_0 }{2} \mathbb{E}_h \left[\dfrac{  2\exp\left(-\beta D \right) \sinh(\beta h_0 + \beta h \Delta )}{1+2 \exp\left(-\beta D  \right) \cosh(\beta h_0 + \beta h \Delta )} \right]^2 ,
\end{align}
where $h$ is now a standard Gaussian variable. \\
We note that, while the Authors did not focus on the annealed approximation of the free energy, the above model --with a Gaussian external field-- has been introduced in \cite{morais2012inverse}\footnote{The Authors focused directly on the RS expression of the quenched free energy, a result that we (re)-obtain in the next Section via Guerra interpolation, see Corollaries \ref{cor:bimodal} (that mirrors Corollary \ref{cor:annealed_bimodal} for the annealed counterpart) and \ref{cor:normal} (that mirrors Corollary \ref{cor:annealed_Gauss} for the annealed counterpart) in Sec. \ref{sec:Guerra}.}.
\end{corollary}

\begin{corollary}
Let $K=0$ and $h_i=H \in \mathbb{R}$ in the Hamiltonian \eqref{eq:Hamiltonian}. This generalization of the Hamiltonian of the GS to a system equipped with spin S  is studied by Katayama and Horiguchi \cite{katayama1999ghatak} and reads as 
    \begin{align}
    \mathcal{H}_{N, D, J, J_0, H}(\bm s| \bm z) := -\dfrac{J}{2\sqrt{N}} \sum_{i,j=1}^N z_{ij} s_i s_j - D \sum_{i=1}^N s_i^2 - H \sum_{i=1}^N s_i - \dfrac{J_0}{2N} \sum_{i,j=1}^N s_i s_j.
    \label{eq:HamiltonianGSS}
\end{align}
    The annealed free energy of the Katayama and Horiguchi model in the  thermodynamic limit $N \to \infty$ reads 
    \begin{align}
            \mathcal{A}^A_{D, J, J_0, H}(\beta) =&\log\left(\SOMMA{s \in \Omega}{}e^{\beta H s-\beta D s^2}\right) +\dfrac{\beta J_0}{2}\left[ \dfrac{\SOMMA{s \in \Omega}{}e^{\beta H s -\beta D s^2}s}{\SOMMA{s \in \Omega}{}e^{\beta H s-\beta D s^2}}\right]^2\notag
    \\
    &+\dfrac{\beta^2 J^2}{4}\left[   \dfrac{\SOMMA{s \in \Omega}{}e^{\beta H s -\beta D s^2}s^2}{\SOMMA{s \in \Omega}{} e^{\beta H s-\beta D s^2}}\right]^2.
    \label{eq:annealed_GSS2}
\end{align}
\end{corollary}
We remind that the above expression of the annealed statistical pressure \eqref{eq:annealed_GSS2} was lacking in the Literature.

\begin{corollary}
    Let $h_i = 0$, for all $i=1, \hdots, N$, $J_0=0$, $J^2=1$ and $S=1$ (i.e. $s_i \in \{-1, 0, +1\}$) in the Hamiltonian model \eqref{eq:Hamiltonian}. In this way, the Hamiltonian reduces to the that of the disordered Blume-Emery-Griffiths-Capel model in the mean-field limit studied by Crisanti and Leuzzi \cite{crisanti2005stable,leuzzi2007spin}, which reads 
\begin{align}
    \mathcal{H}_{N, D, K}(\bm s \vert \bm z) := - \dfrac{1}{2\sqrt{N}} \sum_{i,j=1}^N z_{ij} s_i s_j + D \sum_{i=1}^N s_i^2 - \dfrac{K}{2N} \sum_{i,j=1}^N s_i^2 s_j^2.
    \label{eq:Hamiltonian_BEGC}
\end{align}
    The annealed statistical pressure of this model in the  thermodynamic limit $N \to \infty$ is 
    \begin{align}
    \mathcal{A}^A_{D, K}(\beta) &= \log \left( 1+ 2  e^{-\beta D} \right) + \left[   \dfrac{\beta^2 }{4} \left(\dfrac{e^{-\beta D} }{1+2  e^{-\beta D}}\right)^2+\dfrac{\beta K }{2}\left(\dfrac{e^{-\beta D}}{1+2  e^{-\beta D}}\right)^2\right].
    \label{eq:annealedBEGC2}
\end{align}
We stress that also the above expression was lacking in the Literature.
\end{corollary}

\par\medskip
Now we proceed with the proof of Proposition \ref{prop:annealed}.

\begin{proof}  
    By applying the definition of the annealed statistical pressure provided in \eqref{eq:annealed}   
    we get  explicitly
    \begin{align}
        &\mathcal{A}^A_{N, D, K, J, J_0}(\beta) = \dfrac{1}{N} \mathbb{E}_{\bm h}\log \mathbb{E}_{\bm z} \sum_{\{\bm s\}} \exp \left[ \dfrac{\beta J_0 N}{2} m^2 + \dfrac{\beta J}{2 \sqrt{N}} \sum_{i,j =1}^N z_{ij} s_i s_j - \beta D \sum_{i=1}^N s_i^2 + \beta \sum_{i=1}^N h_i s_i \right. \notag \\
        &\hspace{1cm}\left.+ \dfrac{\beta K}{2N} \left( \sum_{i=1}^N s_i^2 \right)^2  \right] \notag \\
        &= \dfrac{1}{ N} \mathbb{E}_{\bm h}\log \left\{\sum_{\{\bm s\}} \exp \left[\dfrac{\beta J_0 N}{2} m^2  - \beta D \sum_{i=1}^N s_i^2 + \beta \sum_{i=1}^N h_i s_i + \dfrac{\beta K}{2N} \left( \sum_{i=1}^N s_i^2 \right)^2  \right]\right.\notag \\
        &\hspace{1cm}\left.\cdot \mathbb{E}_{\bm z} \exp \left( \dfrac{\beta J}{2 \sqrt{N}} \sum_{i,j =1}^N z_{ij} s_i s_j  \right) \right\}\notag \\
        &= \dfrac{1}{N} \mathbb{E}_{\bm h} \log  \sum_{\{\bm s\}} \exp \left[ \dfrac{\beta J_0}{2N} \left(\sum_{i=1}^N s_i\right)^2 - \beta D \sum_{i=1}^N s_i^2 + \beta \sum_{i=1}^N h_i s_i + \dfrac{\beta K}{2N} \left( \sum_{i=1}^N s_i^2 \right)^2+ \dfrac{\beta^2 J^2}{4 N} \left(\sum_{i=1}^N s_i^2 \right)^2\right]\notag \\
        &= \dfrac{1}{N} \mathbb{E}_{\bm h} \log  \mathbb{E}_{x,y,w} \sum_{\{\bm s\}} \exp \left( \sqrt{\dfrac{\beta J_0 }{ N} }\sum_{i=1}^N s_i x  - \beta D \sum_{i=1}^N s_i^2 + \beta \sum_{i=1}^N h_i s_i  + \dfrac{\beta J}{\sqrt{2N}} \sum_{i=1 }^N s_i^2 y + \sqrt{\dfrac{\beta K}{N}} \sum_{i=1}^N s_i^2 w \right) \notag \\
        &= \dfrac{1}{N} \mathbb{E}_{h} \log  \mathbb{E}_{x} \mathbb{E}_{ y} \mathbb{E}_{w} \prod_{i=1}^N \sum_{s_i} \exp \left[\left( \sqrt{\dfrac{\beta J_0 }{ N} }x  - \beta D s_i + \beta h_i   + \dfrac{\beta J}{\sqrt{2N}} s_i y + \sqrt{\dfrac{\beta K}{N}} w s_i \right) s_i\right] \notag \\
    \end{align}
    where $\mathbb{E}_x,\  \mathbb{E}_w$ and $\mathbb{E}_y$ are Gaussian averages and we make usage of the Gaussian integral
    \begin{align}
        \int_{-\infty}^{+\infty} a e^{-bx^2+cx+d} dx= a \sqrt{\dfrac{\pi}{b}} \exp \left( \dfrac{c^2}{4b} + d\right),
        \label{eq:gaussian}
    \end{align}
    with $a, \ b, \ c, \ d \in \mathbb{R}$. 
    This takes us to the expression reported in thesis for finite size $N$.
    \newline
Now we want to evaluate the thermodynamic limit of the annealed statistical pressure. To do so, we need to analyse the large $N$ behaviour of the argument of the logarithm, therefore, we expand it and, at leading order (for $N \to +\infty$), it reads as 
\begin{align}
    &\left\{\sum_{s \in \Omega}^{} \exp \left[s^2 \left( -\beta D + \dfrac{\beta J}{\sqrt{2N}} w + y\sqrt{\dfrac{\beta K}{
        N}}\right)+s \left(\beta h  + \sqrt{\dfrac{\beta J_0 }{N}}x \right)\right]\right\}^N \notag \\
        &= \exp\left\{ N \log \left[\SOMMA{s \in \Omega}{} \exp\left({\beta hs -\beta D s^2}\right)\left[1+\dfrac{1}{\sqrt{N}}
 \left(w\dfrac{\beta J}{\sqrt{2}} s^2 - 
  x \sqrt{J_0 \beta} s+  y \sqrt{K \beta} s^2\right)\right]+ {o}\left( \dfrac{1}{\sqrt{N}}\right)\right]\right\} \notag \\
\end{align}
Therefore, 
we have at leading order
\begin{align}
    \mathcal{A}^A_{N, D, K, J, J_0}(\beta) =&\mathbb{E}_h\log\left(\SOMMA{s \in \Omega}{}e^{-\beta s^2 D + \beta s h  }\right)  \notag
    \\
    &+\dfrac{1}{N}\mathbb{E}_h \log \mathbb{E}_x  \exp\left[ x\sqrt{N}
   \sqrt{ \beta J_0} \dfrac{\SOMMA{s \in \Omega}{}e^{\beta h s -\beta D s^2}s}{\SOMMA{\gamma=0}{2S} e^{\beta h s -\beta D s^2}}\right]\notag
    \\
    &+\dfrac{1}{N}\mathbb{E}_h \log \mathbb{E}_y \exp\left[ y\sqrt{N}
 \dfrac{\beta J}{\sqrt{2}}  \dfrac{\SOMMA{s \in \Omega}{} e^{\beta h s -\beta D s^2}s^2}{\SOMMA{s \in \Omega}{}e^{\beta h s -\beta D s^2}}\right] \notag 
 \\
 &+\dfrac{1}{N}\mathbb{E}_h \log  \mathbb{E}_w \exp\left[ w\sqrt{N}  \sqrt{\beta K} \dfrac{\SOMMA{s \in \Omega}{} e^{\beta h s -\beta D s^2}s^2}{\SOMMA{s \in \Omega}{}e^{\beta h s -\beta D s^2}}\right].
    \end{align}
Finally, applying the Gaussian averages \eqref{eq:gaussian} in the opposite direction, we manage to elide all $N$ and obtain the  annealed statistical pressure in the thermodynamic limit  as  it appears in the thesis of the Proposition.  
\end{proof}

\section{The quenched free energy}
\label{sec:Guerra}

In this Section we adapt and apply Guerra's interpolation  to study the quenched free energy of the model defined by the Hamiltonian \eqref{eq:Hamiltonian}. With the term {\em Guerra's interpolation} we refer to an ensemble of mathematical techniques  originally developed to tackle spin glasses, see e.g. \cite{guerra1996overlap,guerra_broken,guerra2002thermodynamic},  and later extended also to neural networks, see e.g.  \cite{GuerraNN,AABO-JPA2020, agliari2019dreaming}. This approach is mathematically transparent in every step and, in general, is also less cumbersome and less expensive (in terms of calculations to be performed) than the alternative provided by the celebrated replica trick \cite{MPV}, which continues to preserve historical value. Hence, its percolation within the Disordered Systems and Neural Networks Community is not out of the blue, as it allows to confirm the bulk of outcomes produced in the early days via the other route (and this paper is written exactly with the will of confirming the broad picture of inverse freezing painted by the replica trick). 
\newline
In order to proceed with the computations within the Guerra scheme, we need to make at first a clear assumption on the probability distribution of the order parameters, in particular of the two-replica overlap (see \eqref{Loverlappo} in Definition \ref{def:orderparam}) as the latter may acquire non-trivial behaviour as predicted by Parisi Theory \cite{MPV}.
\newline
Should such a probability distribution $\mathbb{P}(q)$ become a Dirac's delta in the thermodynamic limit, i.e. $\lim_{N \to \infty}\mathbb{P}(q^{a,b}_N)=\delta(q^{a,b}-\bar{q})$, namely the overlap is self-averaging around its mean $\bar{q}$, we recover the RS description that is usually achieved via the replica trick (as having a delta distribution in Guerra's route implies working with an overlap matrix that has all the same off-diagonal entries within the replica trick \cite{B-war4}). 
\newline
Should more peaks be present in the overlap distribution, i.e. introduced a tunable parameter $\theta \in (0,1)$ we write $\lim_{N \to \infty}\mathbb{P}(q^{a,b})=\theta \delta(q^{a,b}-\bar{q}_1)+ (1-\theta) \delta(q^{a,b}-\bar{q}_2)$ with two peaks (at $\bar{q}_1$ and $\bar{q}_2$), we say that replica symmetry is broken (as the Hamming distance $d^{a,b}$ between two generic replicas $a$ and $b$ is no longer unique\footnote{Note that, within a replica symmetric theory, as $d^{a,b} = (1- q^{a,b})/2$, a unique value for $q^{a,b}$ --say $\bar{q}$-- implies a unique value for the Hamming distance $d^{a,b}$, while if two values -say $\bar{q}_1$ and $\bar{q}_2$- are allowed, this is no longer true.} hence different replicas  are no longer equivalent: the larger the number of  peaks, the larger the steps of RSB, up to the Parisi plateau in the limit of infinite steps \cite{MPV,Mourrat,panchenko2013parisi}.
\newline
We face the RS picture in Sec. \ref{ssec:RS} and we deepen the 1-RSB  in Sec. \ref{sec:1-RSB} and we start reminding formally  the definition of the quenched free energy as stated by the next
\begin{definition}
The quenched statistical pressure for the model whose Hamiltonian is \eqref{eq:Hamiltonian} is defined as
    \begin{align}
        \mathcal{A}_{N, D, K, J, J_0}^Q(\beta):= - \beta f^Q_{N, D, K, J, J_0}(\beta)  := \dfrac{1}{N}\mathbb{E} \log  Z_{N, D, K, J, J_0}(\beta \vert \bm z, \bm h),
        \label{eq:quenchedpressure}
    \end{align}
    where $Z_{N, D, K, J, J_0}(\bm z, \bm h)$ is the corresponding partition function \eqref{eq:defBZ} and $f^Q_{N, D, K, J, J_0}(\beta)$ the corresponding quenched free energy.
\end{definition}

\subsection{Replica symmetric scenario}
\label{ssec:RS}
In this section we obtain explicit expressions of the quenched statistical pressure, and the related self-consistency equations for the order parameters, under the assumption of replica symmetry. 
\newline
The main tool we use among Guerra interpolation techniques consists in an interpolation procedure based on the Fundamental Theorem of Calculus \cite{GuerraNN}. The idea in a nutshell is to compare two models, the original one --where nasty interactions among spins are present-- and a one-body model (i.e a Hamiltonian $\mathcal{H} \sim -\sum_{i=1}^N A_i \sigma_i$ for suitable $A_i$) whose statistical mechanical picture is always achievable as its probabilistic structure is intrinsically factorized, i.e. 
$$
\mathbb{P}\left(\sigma_1,\sigma_2,...,\sigma_N\right)=\mathbb{P}(\sigma_1)\mathbb{P}(\sigma_2)...\mathbb{P}(\sigma_N),
$$
where each term $\mathbb{P}(\sigma_i)$ behaves as $\mathbb{P}(\sigma_i) \sim e^{\beta A_i \sigma_i}$.
\newline
In particular, once introduced an interpolating parameter $t \in (0,1)$, we can  define an interpolating quenched statistical pressure $\mathcal{A}_{N, D, K, J, J_0 }^Q(t \vert \beta)$ that coincides with the one pertaining to any of these two models in the two extrema $t=1$ (where the original model is recovered) and $t=0$  (where we are left with the one-body model), i.e.  $\mathcal{A}_{N, D, K, J, J_0}^Q(t=1 \vert \beta)$ is the quantity we want to calculate, while $\mathcal{A}_{N, D, K, J, J_0}^Q(t=0 \vert \beta)$ is the quantity we are able to calculate. 
\newline
These two extrema can be connected as 
$$
\mathcal{A}_{N, D, K, J, J_0}^Q(t=1 \vert \beta) = \mathcal{A}_{N, D, K, J, J_0}^Q(t=0 \vert \beta) + \int_0^1 \frac{d \mathcal{A}_{N, D, K, J, J_0}^Q(t \vert \beta)}{dt}{\Big|_{t=s}}  ds.
$$
By suitably choosing the terms $A_i$, the request that replica symmetry holds then makes the integral analytical in the thermodynamic limit thus allowing to obtain  the explicit solution of the original model.
\newline
To apply this scheme,  we start with the introduction of an interpolating structure, i.e. an interpolating partition function (see \eqref{Zinterpolante}) and its related interpolating quenched statistical pressure (see \eqref{Ainterpolante}) as stated by the next
\begin{definition}\label{def:ZimterpolanteRS}
Let $t \in [0,1]$ and the constants $A, \ B,\ \psi \in \mathbb{R}$ whose specific value will be assigned later on (see \eqref{eq:ABpsi}  and Remark \ref{Rem:Criterion}), consider $z_{ij} \sim \mathcal{N}(0,1)$, for $i<j =1, \hdots, N$ i.i.d. and $z_i \sim \mathcal{N}(0,1)$, for $i=1, \hdots, N$. The interpolating partition function is defined as 
\begin{align}\label{Zinterpolante}
    Z_{N, D, K, J, J_0}(t \vert \beta, \bm z, \bm h) =& \sum_{\{\bm s\}} \exp \left[ \dfrac{\beta J_0 N t}{2} m^2 + \dfrac{\beta J \sqrt{t}}{\sqrt{2N}} \sum_{i<j =1}^N z_{ij} s_i s_j - \beta t D \sum_{i=1}^N s_i^2 + \beta \sum_{i=1}^N h_i s_i   \right.\notag \\
    &\hspace{0.5cm}\left.+\dfrac{\beta K t}{2N} \sum_{i,j=1}^N s_i^2 s_j^2+ A \sqrt{1-t} \sum_{i=1}^N z_i s_i + \dfrac{B}{2}(1-t) \sum_{i=1}^N s_i^2 + (1-t) \psi N m \right],
\end{align}
where we have defined in this case $\bm z = (z_{ij}, z_{i})$ for $i, j=1, \hdots, N$.
\newline
Once we have $Z_{N, D, K, J, J_0}(t \vert \beta, \bm z, \bm h)$, the interpolating quenched statistical pressure simply reads as
\begin{equation}\label{Ainterpolante}
\mathcal{A}_{N, D, K, J, J_0}^Q(t \vert \beta)= \frac{1}{N}\mathbb{E} \log Z_{N, D, K, J, J_0}(t \vert \beta, \bm z, \bm h),
\end{equation}
furthermore $\mathcal{A}_{D, K, J, J_0}^Q(t \vert \beta) = \lim_{N \to \infty}\mathcal{A}_{N, D, K, J, J_0}^Q(t \vert \beta)$.
\newline
Note that, in the two extrema, we recover the original model (i.e. when  $t=1$) or we are left with the one-body model (i.e. when  $t=0$).
\end{definition}
Coherently, we also introduce an interpolating  thermal average $\omega^{(t, \vert \beta, \bm z, \bm h)}_{N, D, K, J, J_0}(\cdot)$ as well as an interpolating quenched average $\langle \cdot \rangle_t$, defined coherently with \eqref{eq:BG} and \eqref{eq:quenched}, as
\begin{align}
    \omega^{(t \vert \beta, \bm z, \bm h)}_{N, D, K, J, J_0} (\cdot):=& \dfrac{1}{Z_{N, D, K, J, J_0}(t \vert \bm z, \bm h)} \sum_{\{\bm s\}} B_{N, D, K, J, J_0}(t \vert \bm z, \bm h) (\cdot), \\
    \langle \cdot \rangle_t :=& \mathbb{E}\left( \omega^{(t \vert \beta, \bm z, \bm h)}_{N, D, K, J, J_0} (\cdot)\right).
\end{align}
From now on we omit the dependence on all the variables of the $t-$Boltzmann average to lighten the notation.\\
Before providing the main theorem of this Section, namely Theorem \ref{thm:GSRS} (and its related proof), we state explicitly the quest to work within a replica symmetric (RS) description, as summarized by the next 
\begin{assumption}
In the RS regime, we assume that the variances of the order parameters $q^{a,b}$, $q^{a,a}$ and $m$ go to zero in the thermodynamic limit, namely we suppose these order parameters to be self-averaging around their means, that we indicate as $\bar{q}$, $Q$ and $\bar{m}$ respectively.  In formulae 
\begin{align}
    \l (q^{a,b}-\bar{q})^2 \r_t \to 0, \quad 
    \l (q^{a,a}-Q)^2 \r_t \to 0, \quad 
    \l (m-\bar{m})^2 \r_t \to 0, \quad N \to + \infty.
\end{align}
\end{assumption}

We can finally state the next
\begin{theorem}
\label{thm:GSRS}
     In the thermodynamic limit, the RS expression of the quenched statistical pressure for the  inverse freezing model defined by the Hamiltonian \eqref{eq:Hamiltonian} reads as
    \begin{align}
        \mathcal{A}^{\mathrm{RS}}_{D, K, J, J_0}(\beta)=& \mathbb{E} \log \sum_{ \bm s \in \Omega} \exp \left[ \beta \left( h + J_0 \m + J \sqrt{\q} z \right) s + \beta\left( \dfrac{\beta J^2}{2} (Q-\q) - D + K Q\right)s^2\right] \notag \\
        &-\dfrac{\beta J_0}{2} \mb^2 - \dfrac{\beta^2 J^2 }{4} \left( Q^2 - \qb^2 \right) -  \dfrac{\beta K}{2} Q^2,
        \label{eq:ARS}
    \end{align}
    where $\mathbb{E}=\mathbb{E}_h \mathbb{E}_z$ and the order parameters must fulfil the constraint $\nabla_{(\bar{q},Q,\bar{m})}\mathcal{A}^{\mathrm{RS}}_{ D, K, J, J_0}(\beta)=0$ that  results in the
    following set of self consistency equations for their evolution in the space of the control parameters: 
    \begin{align}
    Q =& \mathbb{E}\left\{ \dfrac{1}{Z_{D, K, J, J_0}(\beta)} \sum_{\bm s \in \Omega} \exp \left[\beta \left( h + J_0 \m + J \sqrt{\q} z \right) s + \beta\left( \dfrac{\beta J^2}{2} (Q-\q) - D -  K Q\right)s^2\right] s^2  \right\}, \\
    \m =& \mathbb{E}\left\{ \dfrac{1}{Z_{D, K, J, J_0}(\beta)} \sum_{\bm s \in \Omega} \exp \left[ \beta \left( h + J_0 \m + J \sqrt{\q} z \right) s + \beta\left( \dfrac{\beta J^2}{2} (Q-\q) - D -K Q\right)s^2\right] s \right\}, \\
    \q =& \mathbb{E}\left\{ \dfrac{1}{Z_{D, K, J, J_0}(\beta)} \sum_{\bm s \in \Omega} \exp \left[ \beta \left( h + J_0 \m + J \sqrt{\q} z \right) s + \beta\left( \dfrac{\beta J^2}{2} (Q-\q) - D -  K Q\right)s^2\right] s \right\}^2, 
\end{align}
where $Z_{D, K, J, J_0}(\beta)= \sum_{\bm s \in \Omega} \exp \left[ \beta \left( h + J_0 \m + J \sqrt{\q} z \right) s + \beta\left( \dfrac{\beta J^2}{2} (Q-\q) - D +  K Q \right)s^2\right]$.
\end{theorem}
\begin{remark}
\label{rem:RS}
    Also in this case it is possible to rewrite \eqref{eq:ARS} in a more convenient way depending on the value of $S$. \newline 
    If $S \in \mathbb{N}$ we have 
    \begin{align}
         \mathcal{A}^{\mathrm{RS}}_{D, K, J, J_0}(\beta) =& \mathbb{E} \left\{\log \left\{1+\sum_{\gamma=0}^S 2 \exp \left[\beta\left( \dfrac{\beta J^2}{2} (Q-\q) - D -  K Q\right)\left( -1+\dfrac{\gamma}{S}\right)^2\right]\right. \right.\notag \\
         &\left.\left. \cdot\cosh\left[\beta \left( h + J_0 \m + J \sqrt{\q} z \left( -1+\dfrac{\gamma}{S}\right)\right)  \right]\right\}\right\}-\dfrac{\beta J_0}{2} \mb^2 - \dfrac{\beta^2 J^2 }{4} \left( Q^2 - \qb^2 \right) -  \dfrac{\beta K}{2} Q^2. 
    \end{align}
    Instead, for $S \in \mathbb{N}_0 + \dfrac{1}{2}$ we get
    \begin{align}
         \mathcal{A}^{\mathrm{RS}}_{D, K, J, J_0}(\beta) =& \mathbb{E} \left\{\log \left\{\sum_{\gamma=0}^{S-1/2} 2 \exp \left[\beta\left( \dfrac{\beta J^2}{2} (Q-\q) - D -  K Q\right)\left( -1+\dfrac{\gamma}{S }\right)^2\right]\right. \right.\notag \\
         &\left.\left. \cdot\cosh\left[\beta \left( h + J_0 \m + J \sqrt{\q} z \left( -1+\dfrac{\gamma}{S}\right)\right)  \right]\right\}\right\}-\dfrac{\beta J_0}{2} \mb^2 - \dfrac{\beta^2 J^2 }{4} \left( Q^2 - \qb^2 \right) -  \dfrac{\beta K}{2} Q^2. 
    \end{align}
\end{remark}

Theorem \ref{thm:GSRS} reproduces a number of previously obtained results for the various models appeared in the Literature as stated by the several following corollaries that compare in this Section. 
\newline
Let us start with the celebrated SK model as a test case.
\begin{corollary}
\label{cor:SK}
    The 
    RS expression of the quenched statistical pressure of the SK model \eqref{eq:Hamiltonian_SK}, in the thermodynamic limit, reads as  
    \begin{align}
        \mathcal{A}^{\mathrm{RS}}_{J, J_0, H}(\beta) = \mathbb{E}_z \log 2 \cosh \left( \beta H + \beta J_0 \m + \beta J \sqrt{\q} z\right) + \dfrac{\beta^2 J^2}{4} (1-\q)^2 - \dfrac{\beta J_0}{2} \m^2 .
    \end{align}
    where $\m$ and $\q$ fulfill the following self-consistency equations
    \begin{align}
        \m=& \mathbb{E}_z \tanh \left( \beta H + \beta J_0 \m + \beta J \sqrt{\q} z\right), \\
        \q=& \mathbb{E}_z \tanh^2 \left( \beta H + \beta J_0 \m + \beta J \sqrt{\q} z\right).
    \end{align}
We stress that the resulting picture is the same as that painted by the original authors in \cite{sherrington1975solvable} assuming we choose $J_0=0$.
\end{corollary}

\begin{corollary}
\label{cor:GSRS}
 By forcing $K=0$ and $S=1$ in Theorem \ref{thm:GSRS}, the RS expression of the quenched statistical pressure of the GS model \eqref{eq:HamiltonianGS}, in the thermodynamic limit, is obtained as
        \begin{align}
        \mathcal{A}^{\mathrm{RS}}_{D, J, J_0}(\beta)=& \mathbb{E}\log \left[ 1+ 2 \exp \left( \dfrac{\beta^2 J^2}{2} (Q-\qb) - \beta D\right) \cosh \left( \beta h + \beta J_0 \mb + \beta J \sqrt{\q} z\right)\right] \notag \\
        &-\dfrac{\beta J_0}{2} \mb^2 - \dfrac{\beta^2 J^2 }{4} \left( Q^2 - \qb^2 \right),
        \label{eq:ARSGS}
    \end{align}
     where $\mathbb{E}=\mathbb{E}_{h}\mathbb{E}_{z}$ and the order parameters fulfill the following self consistency equations 
    \begin{align}
        \mb =& \mathbb{E} \left[\dfrac{2\exp\left( \dfrac{\beta^2 J^2}{2} (Q-\qb) - \beta D\right) \sinh \left( \beta h + \beta J_0 \mb + \beta J \sqrt{\q} z\right) }{1+ 2 \exp \left( \dfrac{\beta^2 J^2}{2} (Q-\qb) - \beta D\right) \cosh \left( \beta h + \beta J_0 \mb + \beta J \sqrt{\q} z\right)}\right], \label{eq:SCEmRS}\\
        Q=&\mathbb{E} \left[\dfrac{2\exp\left( \dfrac{\beta^2 J^2}{2} (Q-\qb) - \beta D\right) \cosh \left( \beta h + \beta J_0 \mb + \beta J \sqrt{\q} z\right) }{1+ 2 \exp \left( \dfrac{\beta^2 J^2}{2} (Q-\qb) - \beta D\right) \cosh \left( \beta h + \beta J_0 \mb + \beta J \sqrt{\q} z\right)}\right], \\
        \qb =& \mathbb{E} \left[\dfrac{2\exp\left( \dfrac{\beta^2 J^2}{2} (Q-\qb) - \beta D\right) \sinh \left( \beta h + \beta J_0 \mb + \beta J \sqrt{\q} z\right) }{1+ 2 \exp \left( \dfrac{\beta^2 J^2}{2} (Q-\qb) - \beta D\right) \cosh \left( \beta h + \beta J_0 \mb + \beta J \sqrt{\q} z\right)}\right]^2.\label{eq:SCEqRS}
    \end{align}
We stress that the obtained expression is the same provided by Ghatak and Sherrington in \cite{ghatak1977crystal} and by Schupper and Shnerb in \cite{schupper2004spin}.
\end{corollary}
\begin{remark}
The GS model is a rather old  model, hence it has already been treated with more rigorous approaches w.r.t. the replica trick. In particular in 2021 Auffinger and Chen in \cite{auffinger2021thouless} provided an independent proof of the above formula by deriving the Thouless-Anderson-Palmer equations within the cavity field approach. Furthermore, their results cover also the extension provided by Yokota and coworkers, see \cite{da1994first,Yokota1992}. 
\end{remark}
It is interesting to inspect now the role of the external field $h$ in the GS model, in the two cases of a Bernoullian or Gaussian distribution $\mathbb{P}(h)$, as paved in \cite{morais2013inverse} and \cite{morais2012inverse} respectively: the next two Corollaries extend  Corollaries \ref{cor:annealed_bimodal} and \ref{cor:annealed_Gauss} (respectively) from the annealed picture to the quenched  one. 
\begin{corollary}
\label{cor:bimodal}
    Let us now assume that a random external field $h$  is present, whose probability distribution $\mathbb{P}(h)$ reads as 
    \begin{align}
        \mathbb{P}(h)=p \delta(h-h_0)+(1-p) \delta(h+h_0).
        \label{eq:probbernoullian}
    \end{align}
    In the thermodynamic limit, the expression of the RS quenched statistical pressure  of the GS model accordingly becomes 
    \begin{align}
        \mathcal{A}^{\mathrm{RS}}_{D, J, J_0}(\beta)=& p \mathbb{E}_z\log \left[ 1+ 2 \exp \left( \dfrac{\beta^2 J^2}{2} (Q-\qb) - \beta D\right) \cosh \left( \beta h_0 + \beta J_0 \mb + \beta J \sqrt{\q} z\right)\right]\notag \\
        &+(1-p) \mathbb{E}_z\log \left[ 1+ 2 \exp \left( \dfrac{\beta^2 J^2}{2} (Q-\qb) - \beta D\right) \cosh \left( -\beta h_0 + \beta J_0 \mb + \beta J \sqrt{\q} z\right)\right] \notag \\
        &-\dfrac{\beta J_0}{2} \mb^2 - \dfrac{\beta^2 J^2 }{4} \left( Q^2 - \qb^2 \right),
        \label{eq:ARSGS_bimodal}
    \end{align}
    when $\mathbb{E}_z$ is the average with respect to $z$.
\newline
We stress that \eqref{eq:ARSGS_bimodal} is exactly the expression of the quenched statistical pressure provided by the Authors in \cite{morais2013inverse} via the replica trick if we assume that the overlap probability distribution has solely one peak at $\bar{q}$.\footnote{The Authors worked out directly the expression of the quenched free energy under the first step replica symmetry breaking, but it is rather simple to recover the replica symmetric approximation from its first-step RSB generalization, see Section \ref{sec:ATline} later in this manuscript.}
\end{corollary}

The proof of Cor. \ref{cor:bimodal} is made by brute force exploiting the definition of the average for a Bernoullian variable \eqref{eq:probbernoullian}. 

\begin{corollary}
\label{cor:normal}
Let us now assume that an external field $h$  is present, whose probability distribution $\mathbb{P}(h)$ reads as 
    \begin{align}
        \mathbb{P}(h) = \sqrt{\dfrac{1}{2\pi \Delta^2}} \exp \left( -\dfrac{1}{2\Delta^2} (h-h_0)^2\right).
    \end{align}
    In the thermodynamic limit, the expression of the RS quenched statistical pressure  of the GS model becomes
    \begin{align}
        \mathcal{A}^{\mathrm{RS}}_{D, J, J_0, \Delta}(\beta)=& \mathbb{E}_z\log \left[ 1+ 2 \exp \left( \dfrac{\beta^2 J^2}{2} (Q-\qb) - \beta D\right) \cosh \left( \beta J z \sqrt{ \dfrac{\Delta^2}{J^2} + \q}  +\beta J_0 \mb + \beta h_0\right)\right] \notag \\
        &-\dfrac{\beta J_0}{2} \mb^2 - \dfrac{\beta^2 J^2 }{4} \left( Q^2 - \qb^2 \right),
        \label{eq:ARSGS_normal}
    \end{align}
    where $\mathbb{E}_z$ is the Gaussian average with respect to $z$.
    \newline
We stress that \eqref{eq:ARSGS_normal} is exactly the expression of the quenched statistical pressure provided by the Authors in \cite{morais2012inverse} via the replica trick if we assume that the overlap probability distribution has solely one peak\footnote{Again, the (same) Authors worked out directly the expression of the quenched free energy under the first step replica symmetry breaking: see the previous footnote.} at $\bar{q}$.
\end{corollary}

\begin{proof}
The proof of the Corollary is heavily based on the following property of Gaussian averages: let us consider any smooth function $F$ of the Gaussian variables $z$ and $h$. We have 
\begin{align}
    \mathbb{E}_h \mathbb{E}_z \left( F(A z + B h)\right) = \mathbb{E}_z \left( F( \sqrt{A^2 \Delta_1^2 + B^2 \Delta_2^2} z\right),
\end{align}
where $A, B$ are real-valued constants and $\Delta_1$ and $\Delta_2$ are the variances of the Gaussian variables. With the particular choice $A=\beta J \sqrt{\q} $, $B=\beta$ and $\Delta_1=1$ and $\Delta_2=\Delta$ we get the thesis. 
\end{proof}
Let us now move to the Katayama-Horiguchi model \cite{katayama1999ghatak}\footnote{For the sake of correctness, we stress that also Ghatak and Sherrington highlighted the role of general spins $S$ in \cite{ghatak1977crystal}, yet a paper dedicated to the investigation of this extension has been written soon after by  Katayama and Horiguchi \cite{katayama1999ghatak}.}.
\begin{corollary}
\label{cor:GSSRS}
    In the thermodynamic limit, the expression of the RS quenched statistical pressure of the  Katayama-Horiguchi model (namely,  the GS model equipped with spin S) provided by the Hamiltonian  \eqref{eq:HamiltonianGSS}  reads as
\begin{align}
    \mathcal{A}^{\textrm{RS}}_{H, D, J, J_0}(\beta) =& \mathbb{E}_z \log \sum_{ \bm s \in \Omega} \exp \left[ \beta \left( H + J_0 \m + J \sqrt{\q} z \right) s + \beta\left( \dfrac{\beta J^2}{2} (Q-\q) - D \right)s^2\right]\notag \\
    &- \dfrac{\beta^2 J^2}{4} \left( Q^2-\q^2\right) - \dfrac{\beta J_0}{2} \m^2,
    \label{eq:ARSGSS}
\end{align}
where $Q, \ \m, \ \q$ must fulfill the following set of self-consistency equations
\begin{align}\label{eq:self-Q-RS}
    Q =& \mathbb{E}_z\left[ \dfrac{1}{Z} \sum_{s \in \Omega} \exp \left( \beta \left( H + J_0 \m + J \sqrt{\q} z \right) s + \beta\left( \dfrac{\beta J^2}{2} (Q-\q) - D \right)s^2\right) s^2  \right], \\ \label{eq:self-m-RS}
    \m =& \mathbb{E}_z\left[ \dfrac{1}{Z} \sum_{\bm s \in \Omega} \exp \left( \beta \left( H + J_0 \m + J \sqrt{\q} z \right) s + \beta\left( \dfrac{\beta J^2}{2} (Q-\q) - D \right)s^2\right) s \right], \\ \label{eq:self-q-RS}
    \q =& \mathbb{E}_z\left[ \dfrac{1}{Z} \sum_{\bm s \in \Omega} \exp \left( \beta \left( H + J_0 \m + J \sqrt{\q} z \right) s + \beta\left( \dfrac{\beta J^2}{2} (Q-\q) - D \right)s^2\right) s \right]^2
\end{align}
and $Z= \sum_{\bm s \in \Omega} \exp \left[ \beta \left( H + J_0 \m + J \sqrt{\q} z \right) s + \beta\left( \dfrac{\beta J^2}{2} (Q-\q) - D \right)s^2\right]$.
\newline
We stress that the above result was already obtained with the replica trick, see 
\cite{katayama1999ghatak}.
\end{corollary}

\begin{corollary}
\label{cor:BEGCRS} 
Let $h_i = 0$ for all $i=1, \hdots, N$, $J_0=0$, $J^2=1$ and $S=1$ (i.e. $s_i \in \{-1, 0, +1\}$ in \eqref{eq:Hamiltonian}): this way we recover the Hamiltonian of the disordered Blume-Emery-Griffiths-Capel model in the mean field limit investigated by Crisanti and Leuzzi \cite{crisanti2005stable,leuzzi2007spin}. 
\newline
The RS expression of its quenched statistical pressure, in the thermodynamic limit, reads as 
\begin{align}
    \mathcal{A}^{\textrm{RS}}_{D, K}(\beta)= \mathbb{E}_z \log \left[ 1+ 2 \exp \left( \dfrac{\beta^2}{2} (Q-\q) - \beta D + \beta K Q\right) \cosh \left( \beta \sqrt{\q} z\right)\right]-\dfrac{\beta^2}{4} (Q^2 - \q^2) - \dfrac{\beta K}{2} Q^2,
\end{align}
where its order parameters must fulfill the following set of self-consistencies 
\begin{align}
    Q=&  \mathbb{E}_z \left[ \dfrac{2\exp \left( \dfrac{\beta^2}{2} (Q-\q) - \beta D - \beta K Q\right)\cosh \left( \beta \sqrt{\q} z\right) }{1+ 2 \exp \left( \dfrac{\beta^2}{2} (Q-\q) - \beta D -  \beta K Q\right) \cosh \left( \beta \sqrt{\q} z\right)}\right], \\
    \q=& \mathbb{E}_z\left[  \dfrac{2\exp \left( \dfrac{\beta^2}{2} (Q-\q) - \beta D - \beta K Q\right)\sinh \left( \beta \sqrt{\q} z\right) }{1+ 2 \exp \left( \dfrac{\beta^2}{2} (Q-\q) - \beta D - \beta K Q\right) \cosh \left( \beta \sqrt{\q} z\right)}\right]^2.
\end{align}
We stress that the above result was already obtained with the replica trick, see \cite{leuzzi2007spin}.
\end{corollary}

All these corollaries confirm the RS picture that emerged along the years for these models undergoing inverse freezing, yet all these corollaries are simple consequences of Theorem \ref{thm:GSRS} whose proof is so far lacking, hence hereafter we show how to achieve it.
\par\medskip
In order to prove easily Theorem \ref{thm:GSRS}, we need to state some preliminary results as contained in the following two lemmas.  
\begin{lemma}
\label{lemma:1}
The derivative with respect to the interpolating parameter $t$ of the interpolating quenched statistical pressure at finite size $N$ defined in \eqref{Ainterpolante} is 
\begin{align}
    d_t \mathcal{A}_{N,D, K, J, J_0} (t \vert \beta) =& \dfrac{\beta J_0}{2} \l m^2 \r + \dfrac{\beta^2 J^2}{4} \left[ \l q_{11}^2 \r - \l q_{12}^2 \r \right] - \beta D \l q_{11} \r 
 + \beta K \l q_{11}^2 \r \notag \\
 &- \psi \l m \r - \dfrac{A^2}{2} \left[ \l q_{11} \r - \l q_{12} \r \right] - \dfrac{B }{2} \l q_{11} \r .
    \label{eq:dert}
\end{align}
\end{lemma}

The proof is a bit lengthy hence it has been left in the Appendix \ref{app:lemma1} for the sake of clearness.

\begin{lemma}
\label{lemma:dertTDGS}
    If replica symmetry holds, the derivative with respect to $t$ of the interpolating quenched statistical pressure defined in \eqref{Ainterpolante}, in the thermodynamic limit, reduces to  
    \begin{align}
    \label{eq:dtARSTD}
     d_t \mathcal{A}_{D, K, J, J_0} (t \vert \beta) =& - \dfrac{\beta J_0}{2} \mb^2 - \dfrac{\beta^2 J^2 }{4} \left( Q^2 - \qb^2 \right) - \dfrac{\beta K}{2} Q^2.
\end{align}
\end{lemma}

\begin{proof}
In the thermodynamic limit, under the RS assumption, we impose that the variances of all the order parameters go to zero and, by writing this request explicitly, we have 
\begin{align}\
    \l (q_{11} - Q)^2 \r =& \l q_{11}^2 \r + Q^2 - 2Q \l q_{11} \r \to 0, \quad N \to + \infty, \label{eq:var_q11}\\
    \l (q_{12} - \qb)^2 \r =& \l q_{12}^2 \r + \qb^2 - 2\qb \l q_{12} \r \to 0, \quad N \to + \infty, \label{eq:var_q12bis}\\
    \l (m-\mb)^2 \r =& \l m^2 \r + \mb^2 -2\mb \l m \r \to 0, \quad N \to +\infty.\label{eq:var_m}
\end{align}
By replacing \eqref{eq:var_q11}-\eqref{eq:var_m} in \eqref{eq:dert} we can write 
\begin{align}
    d_t \mathcal{A}_{D, K, J, J_0}(\beta\vert t) =& \dfrac{\beta J_0}{2} \left[\l (m-\mb)^2 \r - \mb^2 + 2\mb \l m \r \right] - \beta D \l q_{11} \r - \psi \l m \r- \dfrac{A^2}{2} \left[ \l q_{11} \r - \l q_{12} \r \right]\notag \\
    &+ \dfrac{\beta^2 J^2}{4} \left[ \l (q_{11} - Q)^2 \r- Q^2 + 2Q \l q_{11} \r -\l (q_{12} - \qb)^2 + \qb^2 - 2\qb \l q_{12} \r\right] \notag \\
    & - \dfrac{B }{2} \l q_{11} \r + \dfrac{\beta K}{2} \left[ \l (q_{11} - Q)^2 \r-Q^2 + 2Q \l q_{11} \r\right],
\end{align}
which becomes 
\begin{align}
     d_t \mathcal{A}_{D, K, J, J_0} (t \vert \beta) =& - \dfrac{\beta J_0}{2} \mb^2 - \dfrac{\beta^2 J^2 }{4} Q^2 + \dfrac{\beta^2 J^2}{4} \qb^2 + \dfrac{\beta K}{2} Q^2
\end{align}
if we put 
\begin{align}
    \psi=\beta J_0 \mb, \quad A^2 ={\beta^2 J^2}\qb, \quad B={\beta^2 J^2}(Q-\qb) - 2\beta D - 2 \beta K Q,
    \label{eq:ABpsi}
\end{align}
and impose that the variances \eqref{eq:var_q11}-\eqref{eq:var_m} vanish as $N \to \infty$.
\end{proof}
\begin{remark}\label{Rem:Criterion} 
We comment on the criterion we use to set \textit{a posteriori} the values of the constants $\psi, A, B$ appearing in the interpolative structure \eqref{Zinterpolante} within Definition \ref{def:ZimterpolanteRS}: while mathematically, our aim is to fulfil the concentration of the order parameters on their expected values as prescribed by \eqref{eq:var_q11}-\eqref{eq:var_m}, actually this request confers a physical meaning to the constants as, e.g., by this choice of $A$, the spins experience a rather natural field $\propto \beta J \sqrt{\bar{q}}$. Thus, we replaced the original field acting on $\sigma_i$, that is $\beta\sum_j J_{ij}\sigma_j$, with $\beta \sqrt{\bar{q}}z_i$ and this, in turn, implies that fluctuations of the overlap around its mean $\bar{q}$ vanish as $N \to \infty$.
\end{remark}
We can now proceed with the main proof of this Section.
\begin{proof}
(Of Theorem \ref{thm:GSRS})
Within the present Guerra's interpolation scheme we aim to apply the Fundamental Theorem of Calculus to the interpolating quenched statistical pressure, which -as $N \to \infty$- reads 
\begin{align}
    \mathcal{A}_{D, K, J, J_0}(t=1 \vert \beta) = \mathcal{A}_{D, K, J, J_0}(t=0 \vert \beta) + \int_0^1 d_t \mathcal{A}_{D, K, J, J_0}(t \vert \beta)\Big\vert_{t=s} ds.
    \label{eq:FTC}
\end{align}

From Lemma \ref{lemma:dertTDGS} we get the derivative with respect to $t$ in the thernodynamic limit, which is independent to the interpolating parameter, so the integral is trivial. \\ 
Now the one body term is all we need to evaluate. By direct computations we have
\begin{align}
\label{eq:1BRS}
    \mathcal{A}_{N,D, K, J, J_0}(t=0 \vert \beta) &= \dfrac{1}{N} \mathbb{E}_{h,z} \left[\log \sum_{\{\bm s\} } \exp \left( \beta \sum_{i=1}^N h_i s_i + \psi N m + A \sum_{i=1}^N z_i s_i + B \sum_{i=1}^N s_i^2 \right)\right] \notag \\
    &= \dfrac{1}{N} \sum_{i=1}^N \mathbb{E}_{h,z} \left\{\log \sum_{s_i \in \Omega} \exp \left[ s_i \left( \beta h_i + \psi + A z_i + \dfrac{B}{2} s_i\right)\right] \right\}\notag \\
    &= \mathbb{E}_{h,z} \left[\log \sum_{s \in \Omega} \exp \left(s \left( \beta h + \psi + A z \right)+ \dfrac{B}{2} s^2\right)\right],
\end{align}
where $A, \ B, \ \psi$ have been fixed in \eqref{eq:ABpsi}.
Replacing \eqref{eq:1BRS} and \eqref{eq:dtARSTD} in \eqref{eq:FTC} we get \eqref{eq:ARS}.\\
The expression of the self-consistency equations for the order parameters are simply recovered by brute force, i.e., we need to extremize the quenched statistical pressure with respect to all the order parameters and impose  such extremization to be vanishing\footnote{Strictly speaking we should also study the Hessian of the quenched statistical pressure w.r.t. the order parameters, to be sure that their obtained extreme values do minimize --and not maximize-- the free energy. However, as all these expressions were already available in the Literature, we refer to the original papers to deepen these standard aspects.}: this procedure returns Eqs. \eqref{eq:self-Q-RS} - \eqref{eq:self-q-RS}. 
\end{proof}

\subsection{Broken replica symmetry scenario}
\label{sec:1-RSB}
In spin glass models, self averaging of the order parameters is an oversimplifying assumption in the low temperature limit as the overlap probability distribution $\mathbb{P}(q^{a,b})$ does not concentrate on a unique value $\bar{q}$, i.e. $\lim_{\beta \to \infty}\lim_{N \to \infty}\mathbb{P}(q^{a,b}) \neq \delta (q^{a,b}-\bar{q})$: to improve the description, the simplest generalization of this probability is
\begin{equation}\label{Probability1-RSB}    
\mathbb{P}(q^{a,b}) = \theta \delta (q^{a,b}-\bar{q}_1)+ (1-\theta) \delta (q^{a,b}-\bar{q}_2), \ \ \ \ \theta \in [0,1],
\end{equation}
that assumes that the overlap may concentrate on two values: this is enough to break replica symmetry.  While iterating this reasoning (e.g. by adding more and more peaks in the overlap distribution), eventually a full broken replica symmetric description is achieved \cite{guerra_broken,MPV,talagrand2003spin}, we face solely the 1-RSB in this manuscript as the bulk of existing results for models displaying inverse freezing explored via replica trick cover mainly this scenario only and we aim to recover those formulas via Guerra interpolation hereafter. 
\newline
To accomplish this task we have to generalize the interpolating structure exploited in the previous section in two ways. 
\newline
The former consists in modifying the field experienced by the spins in the one-body model we are left with when $t=0$: rather than plugging a unique Gaussian field $z_i$ modulated by a scalar $A$ as in the interpolating structure \eqref{Zinterpolante}, i.e. a term $\sim \sqrt{1-t}A \sum_{i=1}^N z_i \sigma_i$, now a spin perceives an overall signal that has two contributions\footnote{At the 1-RSB level of description the contributions are two; then,  each further step of RSB adds an extra field acting on the spins.}, i.e. $\sqrt{1-t}A \sum_{i=1}^N z_i \sigma_i \to \sqrt{1-t} \sum_{a=1}^2 A_i^{(a)} \sum_{i=1}^N z_i^{(a)} \sigma_i$.
\newline
This generalization alone would not result particularly relevant (as the sum of Gaussians is still a Gaussian), hence we also need to properly generalize the way we average over these random fields: the proper hierarchical averaging is defined by the expectations \eqref{eqn:Z1} - \eqref{dimenticanza} ({\em vide infra} and Remark \ref{remark10}).
\newline
We can now introduce the interpolating structure we aim to use.
\begin{definition}\label{def:zeta-interpolante-1-RSB}
Let the interpolating parameter $t \in [0,1]$ and  $A^{(1)}, A^{(2)},B,\psi \in \mathbb{R}$ constants to be set a posteriori, $z_{ij} \sim \mathcal{N}(0,1)$, $z_i^{(1)} \sim \mathcal{N}(0,1), z_i^{(2)} \sim \mathcal{N}(0,1)$ i.i.d. We define first
\begin{align}\label{eq:Zinterpolante1-RSB}
    Z_2(t \vert \beta, \bm z, \bm h) =& \sum_{\{\bm s\}} \exp \left[ \dfrac{\beta J_0 N t }{2} m^2 + \dfrac{\beta J \sqrt{t}}{ \sqrt{2N}} \sum_{i,j =1}^N z_{ij} s_i s_j - \beta t D \sum_{i=1}^N s_i^2 + \beta \sum_{i=1}^N h_i s_i + (1-t) \psi N m \right.\notag \\
    &\left.+ \sqrt{1-t} \sum_{a=1}^2 A^{(a)}  \sum_{i=1}^N z_i^{(a)} s_i + (1-t) \dfrac{B}{2} \sum_{i=1}^N s_i^2 +\dfrac{t \beta K}{2N} \sum_{i,j=1}^N s_i^2 s_j^2 \right],
\end{align} 
where $\bm z=(z_{ij}, z^{(a)}_i)$ for $i,j=1, \hdots, N$ and $a=1,2$.
Then, we average out the fields one per time in order to get 
\begin{align}
\label{eqn:Z1}
Z_1( t \vert \beta, \bm z, \bm h)  \coloneqq& \left [\mathbb{E}_2   Z_2(t \vert \beta, \bm z, \bm h) ^\theta \right ]^{1/\theta} \\
\label{eqn:Z0_1-RSBPspin}
Z_0(t \vert \beta, \bm z, \bm h)  \coloneqq&  \exp \left\{\mathbb{E}_1 \left[ \log Z_1(t \vert \beta, \bm z, \bm h)  \right ]\right\} \\ \label{dimenticanza}
{Z}_{N,D, K, J, J_0} ( t \vert \beta, \bm z, \bm h) \coloneqq & Z_0(t \vert \beta, \bm z, \bm h),
\end{align}
where $\mathbb{E}_a, \ a=1,2$ is the average with respect to all the i.i.d. $z_i^{(a)}$.
\end{definition}
We stress that in the construction of the interpolating partition function we have omitted the dependence on $N, D, K, J$ and $J_0$ of the partition function for the sake of clearness.

\begin{definition}\label{def:energialib-interp-1-RSB}
The interpolating quenched statistical pressure is introduced as 
\begin{align}\label{ed:energialib-interp-1-RSB}
    \mathcal{A}_{N,D, K, J, J_0}(t \vert \beta)= \dfrac{1}{N} \mathbb{E}_0 \log Z_{N,D, K, J, J_0}(t \vert \beta, \bm z, \bm h),
\end{align}
where $\mathbb{E}_0$ represents the average with respect to all the i.i.d. $z_{ij}$ and $h_i$.
In the thermodynamic limit, we write
\begin{align}
    \mathcal{A}_{D, K, J, J_0}(t \vert \beta ) = \lim_{N \to + \infty} \mathcal{A}_{N, D, K, J, J_0}(t \vert \beta).
\end{align}
\end{definition}

\begin{remark}\label{remark10}
Following Guerra's prescription \cite{guerra_broken},  given two replicas of the system, we define the following averages
\begin{align}
    &\langle \cdot   \rangle_{t,1}  \coloneqq \mathbb{E}_0\mathbb{E}_1[\mathbb{E}_2 \mathcal{W}_{2,t} \omega_{t}( \cdot)]^2, 
    \label{eq:media1} \\
    &\langle \cdot   \rangle_{t,2}  \coloneqq \mathbb{E}_0\mathbb{E}_1\mathbb{E}_2 \mathcal{W}_{2,t} [\omega_{t}(\cdot) ]^2,\label{eq:media2}
\end{align}    
where 
\begin{equation}
\label{eq:W2}
    \mathcal{W}_{2,t} \coloneqq \dfrac{{Z}_{2}^\theta(t \vert \beta)}{\mathbb{E}_2 {Z}^\theta_{2}(t \vert \beta)}.
\end{equation}
\end{remark}

Also in this case, before providing the main theorem of the Section, namely Theorem \ref{thm:GS1-RSB} (and its proof), we need to premise the following
\begin{assumption}
Within the 1-RSB scheme, we assume that --in the thermodynamic limit-- the overlap concentrates on two possible values only, that we label  as $\bar{q}_1$ and $\bar{q}_2$. We can express this by the quest 
\begin{align}
    \l (q^{a,b} - \q_1)^2 \r_1 \to 0, \quad 
    \l (q^{a,b} - \q_2)^2 \r_2 \to 0, \quad N \to + \infty, \\
    \l (q^{a,a} - Q)^2 \r \to 0, \quad 
    \l (m-\m)^2 \r \to 0, \quad N \to + \infty.
\end{align}
\end{assumption}
We can now state 
\begin{theorem}
\label{thm:GS1-RSB}
The expression of the  1-RSB quenched statistical pressure of the Hamiltonian model \eqref{eq:Hamiltonian} in the thermodynamic limit is 
\begin{align}
    \mathcal{A}_{D, K, J, J_0}^{\textrm{1-RSB}}(\beta) =& \dfrac{1}{\theta}\mathbb{E}_{h}\mathbb{E}_1 \log \mathbb{E}_2\left\{\sum_{\bm s \in \Omega} \exp \left[ \beta \left( h + J_0 \m + J \sqrt{\q_1} z^{(1)} + J \sqrt{\q_2-\q_1} z^{(2)} \right) s \right.\right.\notag \\
    &\left.\left.+ \beta\left( \dfrac{\beta J^2}{2} (Q-\q_2) -D -K Q \right)s^2\right]\right\}^\theta - \dfrac{\beta^2 J^2}{4}\left[ Q^2 - (1-\theta) \q_2^2 - \theta \q_1^2 \right] \notag \\
    &- \dfrac{\beta J_0}{2} \m^2 - \dfrac{\beta K}{2} Q^2,
    \label{eq:A1-RSB}
\end{align}
    where the order parameters must fulfill the following self-consistency equations 
\begin{align}
   \m =& \mathbb{E}_{h}\mathbb{E}_1 \left[ \dfrac{\mathbb{E}_2 \left[ \left( \sum_{\bm s \in \Omega} W_{D, K, J, J_0}(\bm s \vert \beta, \bm z, \bm h)\right)^{\theta-1} \sum_{ \bm s \in \Omega} W_{D, K, J, J_0}(\bm s \vert \beta, \bm z, \bm h) s\right]}{\mathbb{E}_2 \left( \sum_{\bm s \in \Omega} W_{D, K, J, J_0}(\bm s \vert \beta, \bm z, \bm h)\right)^{\theta} }\right], \\
      Q =& \mathbb{E}_{h}\mathbb{E}_1 \left[ \dfrac{\mathbb{E}_2 \left[ \left( \sum_{\bm s \in \Omega} W_{D, K, J, J_0}(\bm s \vert \beta, \bm z, \bm h)\right)^{\theta-1} \sum_{\bm s \in \Omega} W_{D, K, J, J_0}(\bm s \vert \beta, \bm z, \bm h) s^2\right]}{\mathbb{E}_2 \left( \sum_{\bm s \in \Omega} W_{D, K, J, J_0}(\bm s \vert \beta, \bm z, \bm h)\right)^{\theta} }\right], \\
         \q_1 =& \mathbb{E}_{h}\mathbb{E}_1 \left[ \dfrac{\mathbb{E}_2 \left[ \left( \sum_{ \bm s \in \Omega} W_{D, K, J, J_0}(\bm s \vert \beta, \bm z, \bm h)\right)^{\theta-1} \sum_{\bm s \in \Omega} W_{D, K, J, J_0}(\bm s \vert \beta, \bm z, \bm h) s\right]}{\mathbb{E}_2 \left( \sum_{\bm s \in \Omega} W_{D, K, J, J_0}(\bm s \vert \beta, \bm z, \bm h)\right)^{\theta} }\right]^2, \\
            \q_2 =& \mathbb{E}_{ h} \mathbb{E}_1 \left[ \dfrac{\mathbb{E}_2 \left[ \left( \sum_{ \bm s \in \Omega} W_{D, K, J, J_0}(\bm s \vert \beta, \bm z, \bm h)\right)^{\theta-2} \left(\sum_{ \bm s \in \Omega} W_{D, K, J, J_0}(\bm s \vert \beta, \bm z, \bm h) s\right)^2\right]}{\mathbb{E}_2 \left( \sum_{\bm s \in \Omega} W_{D, K, J, J_0}(\bm s \vert \beta, \bm z, \bm h)\right)^{\theta} }\right]
\end{align} 
with \begin{align*}
W_{D, K, J, J_0}(\bm s \vert \beta, \bm z, \bm h) =&  \exp \left[ \beta \left( h + J_0 \m + J \sqrt{\q_1} z^{(1)} + J \sqrt{\q_2-\q_1} z^{(2)} \right) s \right.\notag \\
&\left.+ \beta\left( \dfrac{\beta J^2}{2} (Q-\q_2) - D + K Q\right)s^2\right].\end{align*}
\end{theorem}

\begin{remark}
    Mirroring Remark \ref{rem:RS}, one can rewrite \eqref{eq:A1-RSB} exploiting the role of $S$. If $S \in \mathbb{N}$, \eqref{eq:A1-RSB} becomes
    \begin{align}
        \mathcal{A}_{D, K, J, J_0}^{\textrm{1-RSB}}(\beta) =& \dfrac{1}{\theta}\mathbb{E}_{h}\mathbb{E}_1 \log \mathbb{E}_2\left\{1+\sum_{\gamma=0}^S \cosh \left[ \beta \left( h + J_0 \m + J \sqrt{\q_1} z^{(1)} + J \sqrt{\q_2-\q_1} z^{(2)} \right)\left( -1+\dfrac{\gamma}{S}\right)\right]\right.\notag \\
    &\left.\cdot\exp\left[ \beta\left( \dfrac{\beta J^2}{2} (Q-\q_2) -D -K Q \right)\left( -1+\dfrac{\gamma}{S}\right)^2\right]\right\}^\theta - \dfrac{\beta^2 J^2}{4}\left[ Q^2 - (1-\theta) \q_2^2 - \theta \q_1^2 \right] \notag \\
    &- \dfrac{\beta J_0}{2} \m^2 - \dfrac{\beta K}{2} Q^2.
    \end{align}
    For $S \in \mathbb{N}_0 + \dfrac{1}{2}$ we have 
    \begin{align}
        \mathcal{A}_{D, K, J, J_0}^{\textrm{1-RSB}}(\beta) =& \dfrac{1}{\theta}\mathbb{E}_{h}\mathbb{E}_1 \log \mathbb{E}_2\left\{\sum_{\gamma=0}^{S-1/2} \cosh \left[ \beta \left( h + J_0 \m + J \sqrt{\q_1} z^{(1)} + J \sqrt{\q_2-\q_1} z^{(2)} \right) \left( -1+\dfrac{\gamma}{S}\right) \right]\right.\notag \\
    &\left.\cdot \exp \left[\beta\left( \dfrac{\beta J^2}{2} (Q-\q_2) -D -K Q \right)\left( -1+\dfrac{\gamma}{S}\right)^2\right]\right\}^\theta - \dfrac{\beta^2 J^2}{4}\left[ Q^2 - (1-\theta) \q_2^2 - \theta \q_1^2 \right] \notag \\
    &- \dfrac{\beta J_0}{2} \m^2 - \dfrac{\beta K}{2} Q^2.
    \end{align}
\end{remark}

As we did for the RS inspection, as a first check we re-obtain the expression of the quenched statistical pressure of the SK model hereafter, then we move toward inverse freezing models.
\begin{corollary}
    Mirroring Cor. \ref{cor:SK}, the 1-RSB quenched statistical pressure of the SK model is 
    \begin{align}
        \mathcal{A}^{\mathrm{1-RSB}}_{J, J_0, H}(\beta)= \log 2+ \dfrac{1}{\theta} \mathbb{E}_1 \log \mathbb{E}_2 \left[ \cosh g_{J,H}(\bm z)\right]^\theta -\dfrac{\beta J_0}{2} \mb^2 - \dfrac{\beta^2 J^2}{4} \left[1 + (\theta-1) \qb_2^2 - \theta \qb_1^2\right]
    \end{align}
    where $g_{J,H}(\bm z)= \beta H + \beta J \sqrt{\q_1} z^{(1)} + \beta J \sqrt{\q_2 -\q_1} z^{(2)} + \beta J_0 \m$ and the order parameters must fulfill the following self-consistency equations
    \begin{align}
        \mb =&\mathbb{E}_1 \left[ \dfrac{\mathbb{E}_2 \left[\left( \cosh g_{J,H} (\bm z)\right)^\theta \tanh g_{J,H}(\bm z)\right]}{\mathbb{E}_2 \left( \cosh g_{J,H} (\bm z)\right)^\theta}\right], \\
        \qb_2=& \mathbb{E}_1 \left[ \dfrac{\mathbb{E}_2 \left[\left( \cosh g_{J,H}(\bm z)\right)^\theta \tanh^2 g_{J,H}(\bm z)\right]}{\mathbb{E}_2 \left(\cosh g_{J,H} (\bm z)\right)^\theta}\right], \\
        \qb_1 =& \mathbb{E}_1 \left[ \dfrac{\mathbb{E}_2 \left[\left( \cosh g_{J,H} (\bm z)\right)^\theta \tanh g_{J,H}(\bm z)\right]}{\mathbb{E}_2 \left( \cosh g_{J,H} (\bm z)\right)^\theta}\right]^2.
    \end{align}
    We note that this result correctly returns the (one-step) broken replica symmetry expression of the SK free energy as it appeared in the Literature, see e.g.  \cite{MPV}.
\end{corollary}

We can thus start to inspect how the above expression generalizes at work with models accounting for inverse freezing by tackling the GS test case.  

\begin{corollary}
    Mirroring Corollary \ref{cor:GSRS}, the 1-RSB expression of the quenched statistical pressure of the GS model in the thermodynamic limit is 
\begin{align}
    \mathcal{A}^{\textrm{1-RSB}}_{D, J, J_0}(\beta)=\dfrac{1}{\theta} \mathbb{E}_{h}\mathbb{E}_1 \log \mathbb{E}_2 \left[ 1+ 2 e^\xi  \cosh g_{J, J_0}(\bm z, \bm h)\right]^\theta -\dfrac{\beta J_0}{2} \mb^2 - \dfrac{\beta^2 J^2}{4} \left[Q^2 + (\theta-1) \qb_2^2 - \theta \qb_1^2\right]
    \label{eq:A1-RSBGS}
\end{align}
    where $\xi = \dfrac{\beta^2 J^2}{2}(Q-\q_2) - \beta D$, $g_{J, J_0}(\bm z, \bm h)= \beta h + \beta J \sqrt{\q_1} z^{(1)} + \beta J \sqrt{\q_2 -\q_1} z^{(2)} + \beta J_0 \mb$ and the order parameters must fulfill the following self-consistency equations 
    \begin{align}
        \mb =& \mathbb{E}_h \mathbb{E}_1 \left[ \dfrac{\mathbb{E}_2 \left[\left( 1+ 2 e^\xi \cosh g_{J, J_0} (\bm z, \bm h)\right)^\theta \dfrac{2 e^\xi \sinh g_{J, J_0} (\bm z, \bm h)}{1+ 2 e^\xi \cosh g_{J, J_0} (\bm z, \bm h)}\right]}{\mathbb{E}_2 \left( 1+ 2 e^\xi \cosh g_{J, J_0} (\bm z, \bm h)\right)^\theta}\right], \\
        Q =& \mathbb{E}_h \mathbb{E}_1 \left[ \dfrac{\mathbb{E}_2 \left[\left( 1+ 2 e^\xi \cosh g_{J, J_0} (\bm z, \bm h)\right)^\theta \dfrac{2 e^\xi \cosh g_{J, J_0} (\bm z, \bm h)}{1+ 2 e^\xi \cosh g_{J, J_0} (\bm z, \bm h)}\right]}{\mathbb{E}_2 \left( 1+ 2 e^\xi \cosh g_{J, J_0} (\bm z, \bm h)\right)^\theta}\right], \\
        \qb_2=& \mathbb{E}_h \mathbb{E}_1 \left[ \dfrac{\mathbb{E}_2 \left[\left( 1+ 2 e^\xi \cosh g_{J, J_0} (\bm z, \bm h)\right)^\theta \left(\dfrac{2 e^\xi \sinh g_{J, J_0} (\bm z, \bm h)}{1+ 2 e^\xi \cosh g_{J, J_0} (\bm z, \bm h)}\right)^2\right]}{\mathbb{E}_2 \left( 1+ 2 e^\xi \cosh g_{J, J_0} (\bm z, \bm h)\right)^\theta}\right], \\
        \qb_1 =& \mathbb{E}_h \mathbb{E}_1 \left[ \dfrac{\mathbb{E}_2 \left[\left( 1+ 2 e^\xi \cosh g_{J, J_0} (\bm z, \bm h)\right)^\theta \dfrac{2 e^\xi \sinh g_{J, J_0} (\bm z, \bm h)}{1+ 2 e^\xi \cosh g_{J, J_0} (\bm z, \bm h)}\right]}{\mathbb{E}_2 \left( 1+ 2 e^\xi \cosh g_{J, J_0} (\bm z, \bm h)\right)^\theta}\right]^2. 
    \end{align}
    We note that this result is already present in the Literature, as reported in \cite{ghatak1977crystal}\footnote{We stress that, for the Ghatak-Sherrington model, in the Literature it is available also the broken replica symmetry picture of its fermionic extension, see \cite{QuantumGS}, but we do not deepen any quantum aspect in this paper.}.
\end{corollary}
If we want to consider also the random fields, we obtain the two following corollaries.
\begin{corollary}
    Mirroring Corollary \ref{cor:bimodal}, supposing the probability distribution $\mathbb{P}(h)$ to be 
    \begin{align}
        \mathbb{P}(h)=p \delta(h-h_0)+(1-p) \delta(h+h_0),
    \end{align}
    the expression of the 1-RSB quenched statistical pressure of the GS model \eqref{eq:HamiltonianGS}, in the thermodynamic limit,  becomes
   \begin{align}
    \mathcal{A}^{\textrm{1-RSB}}_{D, J, J_0, p}(\beta)=&\dfrac{p}{\theta} \mathbb{E}_1 \log \mathbb{E}_2 \left[ 1+ 2 e^\xi  \cosh g_{J,J_0,h_0}^{+}(\bm z)\right]^\theta + \dfrac{1-p}{\theta} \mathbb{E}_1 \log \mathbb{E}_2 \left[ 1+ 2 e^\xi  \cosh g_{J,J_0,h_0}^{-}(\bm z)\right]^\theta \notag \\
    &-\dfrac{\beta J_0}{2} \mb^2 - \dfrac{\beta^2 J^2}{4} \left[Q^2 + (\theta-1) \qb_2^2 - \theta \qb_1^2\right]
    \label{eq:A1-RSB_bimodal}
\end{align}
    where $\xi = \dfrac{\beta^2 J^2}{2}(Q-\q_2) - \beta D$, $g_{J,J_0,h_0}^{\pm}(\bm z)= \pm \beta h_0 + \beta J \sqrt{\q_1} z^{(1)} + \beta J \sqrt{\q_2 -\q_1} z^{(2)} + \beta J_0 \mb$.
    \newline
    We note that, also in this case, the above expression of  the 1-RSB quenched statistical pressure returns the expression previously obtained in  \cite{morais2013inverse}, as computed via replica trick.
\end{corollary}

\begin{corollary}
    Mirroring Corollary \ref{cor:normal}, supposing the probability distribution  $\mathbb{P}(h)$ to be
    \begin{align}
         \mathbb{P}(h) = \sqrt{\dfrac{1}{2\pi \Delta^2}} \exp \left( -\dfrac{1}{2\Delta^2} (h-h_0)^2\right),
    \end{align}
    the expression of the 1-RSB quenched statistical pressure of the GS model \eqref{eq:HamiltonianGS}, in the thermodynamic limit,  becomes
 \begin{align}
    \mathcal{A}^{\textrm{1-RSB}}_{D,J, J_0}(\beta)=\dfrac{1}{\theta} \mathbb{E}_1 \log \mathbb{E}_2 \left[ 1+ 2 e^\xi  \cosh g_{J, J_0, h_0, \Delta}(\bm z)\right]^2 -\dfrac{\beta J_0}{2} \mb^2 - \dfrac{\beta^2 J^2}{4} \left[Q^2 + (\theta-1) \qb_2^2 - \theta \qb_1^2\right]
    \label{eq:A1-RSBGS_normal}
\end{align}
    where $\xi = \dfrac{\beta^2 J^2}{2}(Q-\q_2) - \beta D$, $g_{J, J_0, h_0, \Delta}(\bm z)=  \beta J \sqrt{\dfrac{\Delta^2}{J^2} + \q_1} z^{(1)} + \beta J \sqrt{\q_2 -\q_1} z^{(2)} + \beta J_0 \mb + \beta h_0$. 
\newline
    We note that, also in this case, the above expression of  the 1-RSB quenched statistical pressure returns the expression previously obtained in  \cite{morais2012inverse}, as computed via replica trick.
\end{corollary}
\begin{corollary}
    Mirroring Corollary \ref{cor:GSSRS}, focusing on the generalization of the GS model to a system equipped with spin S, as deepened by Katayama and Horiguchi, the expression of its quenched statistical pressure in the thermodynamic limit and under the 1-RSB assumption reads as 
\begin{align}
    \mathcal{A}^{\textrm{1-RSB}}_{D, J, J_0, H}(\beta) =& \dfrac{1}{\theta}\mathbb{E}_1 \log \mathbb{E}_2\left\{\sum_{ \bm s \in \Omega} \exp \left[ \beta \left( H + J_0 \m + J \sqrt{\q_1} z^{(1)} + J \sqrt{\q_2-\q_1} z^{(2)} \right) s \right.\right.\notag \\
    &\left.\left.+ \beta\left( \dfrac{\beta J^2}{2} (Q-\q) + \beta D \right)s^2\right]\right\}^\theta - \dfrac{\beta^2 J^2}{4}\left( Q^2 - (1-\theta) \q_2^2 - \theta \q_1^2 \right) - \dfrac{\beta J_0}{2} \m^2,
\end{align}
where the order parameters must fulfill the following self-consistency equations 
\begin{align}
   \m =& \mathbb{E}_1 \left[ \dfrac{\mathbb{E}_2 \left[ \left( \sum_{ \bm s \in \Omega} W_{D, J, J_0, H}(\bm s \vert \beta)\right)^{\theta-1} \sum_{\bm s \in \Omega} W_{D, J, J_0, H}(\bm s \vert \beta) s\right]}{\mathbb{E}_2 \left( \sum_{ \bm s \in \Omega}W_{D, J, J_0, H}(\bm s \vert \beta)\right)^{\theta} }\right], \\
      Q =& \mathbb{E}_1 \left[ \dfrac{\mathbb{E}_2 \left[ \left( \sum_{\bm s \in \Omega} W_{D, J, J_0, H}(\bm s \vert \beta)\right)^{\theta-1} \sum_{\bm s \in \Omega} W_{D, J, J_0, H}(\bm s \vert \beta) s^2\right]}{\mathbb{E}_2 \left( \sum_{ \bm s \in \Omega} W_{D, J, J_0, H}(\bm s \vert \beta))\right)^{\theta} }\right], \\
         \q_1 =& \mathbb{E}_1 \left[ \dfrac{\mathbb{E}_2 \left[ \left( \sum_{\bm s \in \Omega} W_{D, J, J_0, H}(\bm s \vert \beta)\right)^{\theta-1} \sum_{\bm s \in \Omega} W_{D, J, J_0, H}(\bm s \vert \beta) s\right]}{\mathbb{E}_2 \left( \sum_{ \bm s \in \Omega} W_{D, J, J_0, H}(\bm s \vert \beta)\right)^{\theta} }\right]^2, \\
            \q_2 =& \mathbb{E}_1 \left[ \dfrac{\mathbb{E}_2 \left[ \left( \sum_{ \bm s \in \Omega} W_{D, J, J_0, H}(\bm s \vert \beta)\right)^{\theta-2} \left(\sum_{\bm s \in \Omega} W_{D, J, J_0, H}(\bm s \vert \beta) s\right)^2\right]}{\mathbb{E}_2 \left( \sum_{\bm s \in \Omega} W_{D, J, J_0, H}(\bm s \vert \beta) \right)^{\theta} }\right].
\end{align} 
with \begin{align*}W_{D, J, J_0, H}(\bm s \vert \beta)= & \exp \left[ \beta \left( H + J_0 \m + J \sqrt{\q_1} z^{(1)} + J \sqrt{\q_2-\q_1} z^{(2)} \right) s \right. \notag \\
&\left.+ \beta\left( \dfrac{\beta J^2}{2} (Q-\q) + \beta D \right)s^2\right]. \end{align*}
\newline
We note that, also in this case, the above expression of  the 1-RSB quenched statistical pressure returns the expression previously obtained in \cite{katayama1999ghatak}, as computed via replica trick.
\end{corollary}
\begin{corollary}
    Mirroring Corollary \ref{cor:BEGCRS}, the expression of the quenched statistical pressure of the disordered Blume-Emery-Griffiths-Capel model in the mean field limit deepened by Crisanti and Leuzzi, in the thermodynamic limit and under the 1-RSB assumption, reads as 
\begin{align}
    \mathcal{A}^{\textrm{1-RSB}}_{D, K}(\beta) =&  \dfrac{1}{\theta}\mathbb{E}_1 \log \mathbb{E}_2\left[ 1+ 2\exp \left(\dfrac{\beta^2}{2}  (Q - \q_2) - \beta D - \beta K Q\right) \cosh\left(\beta \sqrt{\q_1} z^{(1)} + \beta \sqrt{\q_2-\q_1} z^{(2)}\right)\right]^\theta \notag \\
    &- \dfrac{\beta^2}{4} \left( Q^2 - (1-\theta) \q^2_2 - \theta \q_1^2 \right), 
\end{align}
whose order parameters must fulfill the following set of self-consistent equations 
\begin{align}
    Q&= \mathbb{E}_1 \left[ \dfrac{\mathbb{E}_2 \left(  1+ 2  e^\xi  \cosh\left(\beta \sqrt{\q_1} z^{(1)} + \beta \sqrt{\q_2-\q_1} z^{(2)}\right) \right)^{\theta-1} 2 e^\xi \cosh  \left(\beta \sqrt{\q_1} z^{(1)} + \beta \sqrt{\q_2-\q_1} z^{(2)}\right)  }{\mathbb{E}_2 \left(  1+ 2 e^\xi  \cosh\left(\beta \sqrt{\q_1} z^{(1)} + \beta \sqrt{\q_2-\q_1} z^{(2)}\right) \right)^\theta}\right], \\
    \q_1 &=  \mathbb{E}_1 \left[ \dfrac{\mathbb{E}_2 \left(  1+ 2  e^\xi  \cosh\left(\beta \sqrt{\q_1} z^{(1)} + \beta \sqrt{\q_2-\q_1} z^{(2)}\right) \right)^{\theta-1} 2 e^\xi \sinh  \left(\beta \sqrt{\q_1} z^{(1)} + \beta \sqrt{\q_2-\q_1} z^{(2)}\right)  }{\mathbb{E}_2 \left(  1+ 2 e^\xi  \cosh\left(\beta \sqrt{\q_1} z^{(1)} + \beta \sqrt{\q_2-\q_1} z^{(2)}\right) \right)^\theta}\right]^2, \\
    \q_2 &=  \mathbb{E}_1 \left[ \dfrac{\mathbb{E}_2 \left(  1+ 2  e^\xi  \cosh\left(\beta \sqrt{\q_1} z^{(1)} + \beta \sqrt{\q_2-\q_1} z^{(2)}\right) \right)^{\theta-2} \left(2 e^\xi \sinh  \left(\beta \sqrt{\q_1} z^{(1)} + \beta \sqrt{\q_2-\q_1} z^{(2)}\right)^2\right)  }{\mathbb{E}_2 \left(  1+ 2 e^\xi  \cosh\left(\beta \sqrt{\q_1} z^{(1)} + \beta \sqrt{\q_2-\q_1} z^{(2)}\right) \right)^\theta}\right],
\end{align}
\normalsize
with $\xi = \left[\dfrac{\beta^2}{2}  (Q - \q_2) - \beta D + \beta K Q\right]$.
\newline
We note that, also in this case, the above expression of  the 1-RSB quenched statistical pressure returns the expression previously obtained in \cite{leuzzi2007spin}, as computed via replica trick.
\end{corollary}

Also in this case, to prove Theorem  \ref{thm:GS1-RSB}, we want to exploit the Fundamental Theorem of Calculus \eqref{eq:FTC}. Therefore, mirroring the previous Section, we state the following Lemmas as premises. 
\begin{lemma}  
\label{lemma:3}
At finite size $N$, the derivative with respect to $t$ of the interpolating quenched statistical pressure \eqref{ed:energialib-interp-1-RSB},  as introduced by Definition \ref{def:zeta-interpolante-1-RSB}, reads 
\begin{align}
    &d_t \mathcal{A}_{N, D, K, J, J_0}(t \vert \beta) = \dfrac{\beta J_0}{2} \l m^2 \r + \dfrac{\beta^2 J^2}{4} \left[ \l q_{11}^2 \r - (1-\theta) \l q_{12}^2 \r_2 - \theta \l q_{12}^2 \r_1 \right] - \beta D \l q_{11} \r - \psi \l m \r - \dfrac{\beta K}{2} \l q_{11}^2 \r \notag \\
    &- \dfrac{(A^{(1)})^2}{2} \left[ \l q_{11} \r - (1-\theta) \l q_{12} \r_2 - \theta \l q_{12} \r_1 \right] - \dfrac{(A^{(2)})^2}{2} \left[ \l q_{11} \r - (1-\theta)\l q_{12} \r_2 \right] - \dfrac{B }{2} \l q_{11} \r.
    \label{eq:dert1-RSBN}
\end{align}
\end{lemma}

Since the proof of this Lemma is cumbersome in computations we report it in Appendix \ref{app:lemma3}.

\begin{lemma}
    The derivative with respect to $t$ of the interpolating quenched statistical pressure in the thermodynamic limit is 
    \begin{align}
    \label{eq:dert1-RSB}
    d_t \mathcal{A}_{D, K, J, J_0}(t \vert \beta)=-\dfrac{\beta J_0}{2} \mb^2 - \dfrac{\beta^2 J^2}{4} \left[ Q^2 + (\theta-1) \qb_2^2 - \theta \qb_1^2\right] - \dfrac{\beta K}{2} Q^2.
\end{align}
\end{lemma}

\begin{proof}
In order to get the expression  \eqref{eq:dert1-RSBN} in the thermodynamic limit we exploit that, within the 1-RSB assumption,  
\begin{align}
    \l (q_{11} - Q)^2 \r =& \l q_{11}^2 \r + Q^2 - 2Q \l q_{11} \r \to 0, \quad N \to + \infty, \label{eq:var_q111-RSB}\\
    \l (q_{12} - \qb_a)^2 \r_a =& \l q_{12}^2 \r + \qb_a^2 - 2\qb_a \l q_{12} \r_a \to 0, \quad N \to + \infty, \quad a=1, 2 \\
    \l (m-\mb)^2 \r =& \l m^2 \r + \mb^2 -2\mb \l m \r \to 0, \quad N \to +\infty. \label{eq:var_m1-RSB}
\end{align}

Therefore, by replacing \eqref{eq:var_q111-RSB}-\eqref{eq:var_m1-RSB} in \eqref{eq:dert1-RSBN}, we have
\begin{align}
    d_t \mathcal{A}_{D, K, J, J_0}(t \vert \beta)=-\dfrac{\beta J_0}{2} \mb^2 - \dfrac{\beta^2 J^2}{4} \left( Q^2 + (\theta-1) \qb_2^2 - \theta \qb_1^2\right) - \dfrac{\beta K}{2}Q^2
\end{align}
where we imposed that the variances vanish and we put
\begin{align}
\label{eq:constant1-RSB}
    &\psi = \beta J_0 \mb, , \quad B= \dfrac{\beta^2 J^2}{2} (Q-\q_2) - 2\beta D +2\beta K Q \notag \\
    &(A^{(1)})^2 = \beta^2 J^2 \qb_1, \quad (A^{(2)})^2 = \beta^2 J^2 (\q_2-\q_1).
\end{align}
\end{proof}
\begin{remark}
We comment on the criterion we use to set a posteriori the values of the constants $\psi, A_1, A_2, B$ appearing in the interpolative structure \eqref{eq:Zinterpolante1-RSB} within Definition \ref{def:zeta-interpolante-1-RSB}: while mathematically, our aim is to fulfill the concentration of the order parameters on the allowed expected values  as prescribed by \eqref{eq:var_q111-RSB}-\eqref{eq:var_m1-RSB}, actually this confers a physical meaning to the constants as, e.g., by this choice of $A^{(1)}$ and $A^{(2)}$, the spins experience two rather natural fields, the former $\propto \beta J \sqrt{\bar{q}_1}$, the latter $\propto \beta J \sqrt{(\bar{q}_2-\bar{q}_1)}$. 
\end{remark}
\begin{proof} (Of Theorem \ref{thm:GS1-RSB})
In order to apply the Fundamental Theorem of Calculus \eqref{eq:FTC}, we only need to evaluate the one-body term $\mathcal{A}_{N, D, K, J, J_0}(t=0 \vert \beta)$, that reads as 
\small
\begin{align}
    \mathcal{A}_{N, D, K, J, J_0}(t=0 \vert \beta ) &= \dfrac{1}{\theta N} \mathbb{E}_{\bm h}\mathbb{E}_1 \log \mathbb{E}_2 \left[ \sum_{\{\bm s\}} \exp \left( \psi \sum_{i=1}^N s_i + \sum_{a=1}^2 A^{(a)} \sum_{i=1}^N z_i^{(a)} s_i + \dfrac{B}{2} \sum_{i=1}^N s_i^2 + \beta \sum_{i=1}^N h_i s_i \right)\right]^{\theta} \notag \\ 
    &= \dfrac{1}{\theta N} \mathbb{E}_{\bm h}\mathbb{E}_1 \log \mathbb{E}_2 \left\{ \prod_{i=1}^N \sum_{s_i \in \Omega} \exp\left[s_i \left( \psi + \sum_{a=1}^2 A^{(a)} z_i^{(a)} + \dfrac{B}{2} s_i + \beta h_i\right)\right]\right\} \notag \\ 
     &= \dfrac{1}{\theta } \mathbb{E}_{\bm h}\mathbb{E}_1 \log \mathbb{E}_2 \left\{\sum_{s \in \Omega} \exp\left[s\left( \psi + \sum_{a=1}^2 A^{(a)} z_i^{(a)} + \dfrac{B}{2} s + \beta h_i\right)\right]\right\}
    \label{eq:1B1-RSB}
\end{align}
\normalsize
where the values of the constants are set in \eqref{eq:constant1-RSB}.

By putting together \eqref{eq:dert1-RSB} and \eqref{eq:1B1-RSB} in the Fundamental Theorem of Calculus we reach the thesis. 
\end{proof}

\begin{remark}
\label{remark:A1-RSBRS}
    We notice that both for $\theta=0$ and $\theta=1$ the 1-RSB quenched statistical pressure \eqref{eq:A1-RSB} reduces to the RS one \eqref{eq:ARS}. By way of example, for $\theta=1$ we have
\begin{align}
    \mathcal{A}^{\textrm{1-RSB}}_{D, K, J, J_0}&(\beta, \theta =1)=-\dfrac{\beta J_0}{2} \mb^2 - \dfrac{\beta^2 J^2 }{4} \left( Q^2 - \q_1^2\right) - \dfrac{\beta K}{2}Q^2\notag \\
    &+ \mathbb{E}_h\mathbb{E}_1 \log \mathbb{E}_2 \sum_{s \in \Omega} \exp \left[ \beta \left( h + J_0 \m + J \sqrt{\q_1} z^{(1)} + J \sqrt{\q_2-\q_1} z^{(2)}\right) s \right. \notag \\ 
    &\left.\hspace{1.5cm}+ \beta \left( \dfrac{\beta J^2 }{2}(Q-\q_2) - D - KQ\right) s^2\right]\notag \\
    =&-\dfrac{\beta J_0}{2} \mb^2 - \dfrac{\beta^2 J^2 }{4} \left( Q^2 - \q_1^2\right) - \dfrac{\beta K}{2}Q^2\notag \\
    &+ \mathbb{E}_h\mathbb{E}_1 \log \sum_{s \in \Omega} \exp \left[ \beta \left( h + J_0 \m + J \sqrt{\q_1} z^{(1)} \right) s + \beta \left( \dfrac{\beta J^2 }{2}(Q-\q_2) - D - KQ\right) s^2\right] \notag \\
    &\hspace{1.5cm}\cdot\mathbb{E}_2 \left[\exp \left(  \beta J \sqrt{\q_2-\q_1} z^{(2)} s\right)\right]\notag \\
    =& -\dfrac{\beta J_0}{2} \mb^2 - \dfrac{\beta^2 J^2 }{4} \left( Q^2 - \q_1^2\right) - \dfrac{\beta K}{2}Q^2\notag \\
    &+ \mathbb{E}_h\mathbb{E}_1 \log \sum_{s \in \Omega} \exp \left[ \beta \left( h + J_0 \m + J \sqrt{\q_1} z^{(1)} \right) s + \beta \left( \dfrac{\beta J^2 }{2}(Q-\q_2) - D - KQ\right) s^2 \right. \notag \\ 
    &\left.\hspace{1.5cm}+ \left(  \beta J \sqrt{\q_2-\q_1} s\right)^2\right] \notag \\
    =&-\dfrac{\beta J_0}{2} \mb^2 - \dfrac{\beta^2 J^2 }{4} \left( Q^2 - \q_1^2\right) - \dfrac{\beta K}{2}Q^2\notag \\
    &+ \mathbb{E}_h\mathbb{E}_1 \log \sum_{s \in \Omega} \exp \left[ \beta \left( h + J_0 \m + J \sqrt{\q_1} z^{(1)} \right) s + \beta \left( \dfrac{\beta J^2 }{2}(Q-\q_1) - D - KQ\right) s^2\right] \notag \\
    \equiv& \mathcal{A}^{\textrm{RS}}_{D, K, J, J_0}(\beta).
\end{align}
\normalsize
\end{remark}

\subsection{Instability of the RS picture}\label{sec:ATline}
In this Section we wonder if there is a way to understand systematically when the RS approximation of the quenched free energy becomes unstable and the expression achieved under the 1-RSB preferable. Historically this question has been faced firstly by de Almeida-Thouless \cite{de1978stability} at work with the SK model \cite{sherrington1975solvable}, the harmonic oscillator of spin glasses. Since then, the line (or surface, hyper-surface, etc. depending on the model and how many control parameters are required to describe its behavior) is called {\em AT-line} and, in this last section of computations, we do search for {\em AT-lines} for these models undergoing inverse freezing.
\newline
Over the years, several approaches have been proposed in the Literature to accomplish this task (see, e.g. \cite{chen2021almeida,toninelli2002almeida}) and we rely upon the technique presented in \cite{albanese2023almeida} that, in a nutshell, consists in expanding the 1-RSB free energy in terms of the RS one.  This route, that will be paved in this section, ends up to the result summarized in the next
\begin{theorem}
The RS expression of the quenched free energy related to the class of models undergoing inverse freezing accounted by Hamiltonian \eqref{eq:Hamiltonian} is valid until the following  {\em AT-line}\footnote{Strictly speaking, the boundary is obtained by requesting the equality --rather than the inequality-- in the following equation, yet --historically \cite{de1978stability}-- it has been represented as a bound, hence we preserve this notation.} in the space of the control parameters 
\begin{align}
    &1-\beta^2 J^2 \mathbb{E}_h \mathbb{E}_1\left\{ \dfrac{\left(\sum_{ s \in \Omega} \tilde W_{D, K, J, J_0}(s \vert \beta, \bm z, \bm h) s\right)^2  }{\left(\sum_{s \in \Omega} \tilde W_{D, K, J, J_0}(s \vert \beta, \bm z, \bm h))\right)^2}-\dfrac{\left(\sum_{s \in \Omega} \tilde W_{D, K, J, J_0}( s \vert \beta, \bm z, \bm h) s^2\right)}{\left(\sum_{s \in \Omega} \tilde W_{D, K, J, J_0}(s \vert \beta, \bm z, \bm h))\right)} \right\}^2 <0,
    \label{eq:AT_thm}
\end{align}
where $ \tilde W_{D, K, J, J_0}(s \vert \beta, \bm z, \bm h) = \exp \left[ \beta \left( h + J_0 \m + J \sqrt{\q} z^{(1)}\right) s + \beta \left( \dfrac{\beta J^2}{2} (Q-\q) - D+ KQ\right)s^2\right] $ and the order parameters $Q, \ \m, \q$ fulfill the self-consistency equations in \eqref{eq:SCEmRS}-\eqref{eq:SCEqRS}.
\end{theorem}
We stress that the above Theorem, in order to hold for all the cases we are dealing with, returns the boundary line in a rather implicit way, nevertheless in the specific corollaries that follow --where we analyze the single models one by one-- the AT-line can be made explicit.
\begin{proof}
Our method is based on the observation that, for $\theta \to 0$ (or, equivalently, for $\theta \to 1$),  one of the two delta peaks in the 1-RSB probability distribution \eqref{Probability1-RSB} vanishes and the whole 1-RSB expression of the quenched free energy naturally collapses into the RS one. 
\newline
Here, with no loss of generality \cite{albanese2023almeida}, we focus on the $\theta \to 1$ case to we prove that for values of $\theta$ close but away from one, the 1-RSB expression of the quenched  statistical pressure  is greater than the corresponding RS expression, i.e. 
$\mathcal{A}^{1-RSB}(\m,\bar{q}_{2},\bar{q}_1|\theta)>\mathcal{A}^{RS}(\m,\bar{q})$, 
below the \textit{AT line}. Indeed, remembering that the free energy is opposite in sign w.r.t. the statistical pressure, this observation implies that in the control parameter space, below that threshold the broken replica expression should be preferred. In order to lighten the notation, we omit from now on the explicit dependence on ${D, K, J, J_0}$ and $\beta$ on the quenched statistical pressure and, instead, we stress the presence of $\theta$ and the order parameters $\m, \ Q, \ \q_1,\  \q_2$. 
\newline
First  we expand the 1-RSB quenched statistical pressure \eqref{eq:A1-RSB} around $\theta=1$ up to the first order to write
\begin{align}
\label{eq:expansionAT}
    \mathcal{A}^{\textrm{1-RSB}}(\theta \vert \m, Q, \qb_1, \qb_2)&=\mathcal{A}^{\textrm{1-RSB}}(\theta=1 \vert \m, Q, \qb_1, \qb_2) + (\theta-1) d_{\theta}\mathcal{A}^{\textrm{1-RSB}}(\theta \vert \m, Q, \qb_1, \qb_2)\Big \vert_{\theta=1}\notag \\
    &=\mathcal{A}^{\textrm{RS}} + (\theta-1) d_{\theta}\mathcal{A}^{\textrm{1-RSB}}(\theta \vert \m, Q, \qb_1, \qb_2)\Big \vert_{\theta=1}.
\end{align}
The last passage is justified by Remark \ref{remark:A1-RSBRS}.

Since  also the self-consistency equations of the order parameters depends on $\theta$, we need to expand them too: we have,  up to the first order, 
\begin{align}
    \mb =& \mb + (\theta-1) \mathcal{A}(\mb, Q, \qb_1, \qb_2), \quad \quad
    Q = Q + (\theta-1) \mathcal{B}(\mb, Q, \qb_1, \qb_2) \notag\\
    \q_1 =& \q + (\theta-1) \mathcal{C}(\mb, Q, \qb_1, \qb_2), \quad \quad
    \q_2 = \tilde{q}_2 + (\theta-1) \mathcal{E}(\mb, Q, \qb_1, \qb_2).
\end{align}
where the coefficient $\mathcal{A}, \mathcal{B}, \mathcal{C}$ and $\mathcal{E}$ are real-valued functions of the control and order parameters, whose explicit expressions are reported in Appendix \ref{app:ABC}.  


Now we compute the derivative of the 1-RSB expression of the quenched statistical pressure with respect to $\theta$ and evaluate the result at $\theta=1$. For our purposes, indeed, we are interested on the sign of this term (in this way we have a strict inequality between RS and 1-RSB approximations of the quenched statistical pressure). 
Since the computations are lengthy but straightforward, we report it in Appendix \ref{app:K}.
Then, we have
\begin{align}
    K(\mb, \q, \tilde q_2)=&d_{\theta}\mathcal{A}^{\textrm{1-RSB}}\Big \vert_{\theta=1}=-\dfrac{\beta^2 J^2}{4}\left( (\tilde{q}_2^2 - \q^2 \right)- \mathbb{E}_h \mathbb{E}_1 \log \left[ {\mathbb{E}_2  \left(\sum_{s \in \Omega} \mathbb W_{D, K, J, J_0}(\bm s \vert \beta, \bm z, \bm h) \right)}\right] \notag \\
    &+ \mathbb{E}_h\mathbb{E}_1 \left[ \dfrac{\mathbb{E}_2 \left[ \left(\sum_{s \in \Omega} \mathbb W_{D, K, J, J_0}(\bm s \vert \beta, \bm z, \bm h) \right) \log \left(\sum_{s \in \Omega} \mathbb W_{D, K, J, J_0}(\bm s \vert \beta, \bm z, \bm h) \right)\right]}{\mathbb{E}_2 \left(\sum_{s \in \Omega} \mathbb W_{D, K, J, J_0}(\bm s \vert \beta, \bm z, \bm h) \right)}\right] 
    \label{eq:K}
\end{align}
where \begin{align*}\mathbb W_{D, K, J, J_0}( s \vert \beta, \bm z, \bm h)=\exp &\left[ \beta (h + J_0 \m + J \sqrt{\q} z^{(1)} + J \sqrt{\tilde q_2 - \q} z^{(2)})s \right.\notag \\
&\left.+ \beta \left(\dfrac{\beta J^2}{2}(Q- \tilde q_2) - D -K Q \right)s^2\right].\end{align*} We emphasise that, although it is necessary to expand the order parameters as they too depend on $\theta$, the corrections have played no role when we consider the derivative computed at $\theta=1$.
\newline
The problem now is that the sign of $K(\m, \q, \tilde q_2)$ is not defined. However, note that as $\tilde q_2 \to \q$, we have that $K(\mb, \q, \q)=0$. So, our purpose is to find the extremum of the function $K(\mb, \q, x)$, with $x \in [0, \q]$ and $\mb, \ \q$ fixed, by requesting
\begin{align}
    \partial_x K(\mb, \q,x)=& -\dfrac{\beta^2 J^2}{2}\mathbb{E}_h\mathbb{E}_1 \left\{ \dfrac{1}{\mathbb{E}_2 \left[\sum_{s \in \Omega} \mathbb W_{D, K, J, J_0}(s \vert \beta, \bm z, \bm h) \right]}\mathbb{E}_2 \left[ \dfrac{\left(\sum_{s \in \Omega} \mathbb W_{D, K, J, J_0}(s \vert \beta, \bm z, \bm h) s\right)^2}{\left(\sum_{s \in \Omega} \mathbb W_{D, K, J, J_0}(s \vert \beta, \bm z, \bm h) \right)}\right]\right\}\notag \\
    &+\dfrac{\beta^2 J^2}{2} x =0
\end{align}
which means that $x \equiv \tilde q_2$.

Since $K(\m, \q, \q)=0$, if the extremum $x= \tilde q_2$ is global in the domain $[0,\q]$, then, when $x=\tilde q_2$  is a minimum, $K(\m, \q, \tilde q_2) < 0$.
Therefore, computing the second derivative of $K(\mb, \q, x)$ with respect to $x$ we get
\begin{align}
\partial_{x^2} K(\mb, \q, x)\vert_{x=\tilde q_2}=& -\dfrac{\beta^2 J^2}{2} \left\{ 1-\beta^2 J^2 \mathbb{E}_h\mathbb{E}_1 \left[ \dfrac{1}{\mathbb{E}_2 \left(\sum_{s \in \Omega}\mathbb W_{D, K, J, J_0}(\bm s \vert \beta, \bm z, \bm h) \right)} \right. \right.\notag \\
&\left. \left.\cdot\mathbb{E}_2 \left[ \dfrac{ \left(\sum_{s \in \Omega} \mathbb W_{D, K, J, J_0}(\bm s \vert \beta, \bm z, \bm h) s\right)^2 }{\left(\sum_{s \in \Omega} \mathbb W_{D, K, J, J_0}(\bm s \vert \beta, \bm z, \bm h)\right)^{3/2}}-\dfrac{ \left(\sum_{s \in \Omega} \mathbb W_{D, K, J, J_0}(\bm s \vert \beta, \bm z, \bm h) s^2\right)}{\left(\sum_{s \in \Omega} \mathbb W_{D, K, J, J_0}(\bm s \vert \beta, \bm z, \bm h)\right)^{1/2}}\right]^2\right] \right\},
\end{align}
which returns the condition
\begin{align}
     1-\beta^2 J^2 &\mathbb{E}_h\mathbb{E}_1 \left[ \dfrac{1}{\mathbb{E}_2 \left(\sum_{s \in \Omega}\mathbb W_{D, K, J, J_0}(\bm s \vert \beta, \bm z, \bm h) \right)} \right.  \notag \\
& \left.\cdot\mathbb{E}_2 \left[ \dfrac{ \left(\sum_{s \in \Omega} \mathbb W_{D, K, J, J_0}(\bm s \vert \beta, \bm z, \bm h) s\right)^2 }{\left(\sum_{s \in \Omega} \mathbb W_{D, K, J, J_0}(\bm s \vert \beta, \bm z, \bm h)\right)^{3/2}}-\dfrac{ \left(\sum_{s \in \Omega} \mathbb W_{D, K, J, J_0}(\bm s \vert \beta, \bm z, \bm h) s^2\right)}{\left(\sum_{s \in \Omega} \mathbb W_{D, K, J, J_0}(\bm s \vert \beta, \bm z, \bm h)\right)^{1/2}}\right]^2\right]  <0.
\end{align}
The condition equivalent to the AT-line for the SK model, reported in the thesis, is found when $\tilde q_2 \to \q$.
\end{proof}
Now, in order to recover explicit expressions of this boundary for all the variations on theme on inverse freezing we are discussing, mirroring Corollaries \ref{cor:GSRS}-\ref{cor:BEGCRS}, we need simply to make explicit assumptions on the values of $h, \ S, \ K$ and $D$.

As we did for the whole manuscript, the first test case is not for inverse freezing, rather it is the standard SK model we use as reference. 

\begin{corollary}
 If $K=D=0$ and $S=1/2$, such that the Hamiltonian \eqref{eq:Hamiltonian} turns into the SK-one \eqref{eq:Hamiltonian_SK}, we get the following AT-line
    \begin{align}
        1- \beta^2 J^2 \mathbb{E}_1 \left\{ \cosh^4 (\beta (H + J_0 \m + J \sqrt{\q} z^{(1)})) \right\} <0,
    \end{align}
    that coincides with the original de Almeida and Thouless's AT-line \cite{de1978stability}.
\end{corollary}
Now we move to inspect models for inverse freezing.
\begin{corollary}
For the GS model (i.e. by choosing $K=0$ and $S=1$ in the general Hamiltonian \eqref{eq:Hamiltonian} so to re-obtain the expression provided in \eqref{eq:HamiltonianGS})  we have that 
    \begin{align}
        \sum_{s \in \Omega} \tilde W_{D, K, J, J_0}(\bm s \vert \beta, \bm z, \bm h) =& 1+ 2 e^\xi \cosh g_{J_0, J}(\beta, \bm z, \bm h), \\
        \sum_{s \in \Omega} \tilde W_{D, K, J, J_0}(\bm s \vert \beta, \bm z, \bm h) s =& 2 e^\xi \sinh g_{J_0, J}(\beta, \bm z, \bm h), \\
        \sum_{s \in \Omega} \tilde W_{D, K, J, J_0}(\bm s \vert \beta, \bm z, \bm h) s^2 =& 2 e^\xi \cosh g_{J_0, J}(\beta, \bm z, \bm h),
    \end{align}
    where $\xi = \dfrac{\beta^2 J^2 }{2} (Q-\q) - \beta D -  \beta K Q$ and $g(\beta, J)= \beta(h + J_0 \m + J \sqrt{\q}z^{(1)})$.
    \newline
    Thus, the AT-line for the Ghatak-Sherrington model reads as 
    \begin{align}
        1- {\beta^2 J^2} \mathbb{E}_h\mathbb{E}_1 \left\{ \dfrac{2 e^\xi (2e^\xi + \cosh g_{J_0, J}(\beta, \bm z, \bm h))}{(1+ 2e^\xi \cosh g_{J_0, J}(\beta, \bm z, \bm h))^2}\right\}^2 <0.
    \end{align}
    We point out that we obtained the same expression provided by the original Authors  in \cite{ghatak1977crystal} (achieved via the standard route \cite{de1978stability}, that was also the unique route at that time). 
\end{corollary}
\begin{corollary}
For the Katayama and Horiguchi generalization of the  GS model (i.e.  by selecting $K=0$, $h=H \in \mathbb{R}$ and general spin S in the Hamiltonian \eqref{eq:Hamiltonian}, so to recover their Hamiltonian \eqref{eq:HamiltonianGSS}), we have 
\begin{align}
    &1-\beta^2 J^2 \mathbb{E}_1\left\{ \dfrac{\left(\sum_{s \in \Omega} \tilde W_{D, K, J, J_0, H}(\bm s \vert \beta, \bm z) s\right)^2 }{\left(\sum_{s \in \Omega} \tilde W_{D, K, J, J_0, H}(\bm s \vert \beta, \bm z)\right)^2} -\dfrac{\left(\sum_{s \in \Omega} \tilde W_{D, K, J, J_0, H}(\bm s \vert \beta, \bm z) s^2\right) }{\left(\sum_{s \in \Omega} \tilde W_{D, K, J, J_0, H}(\bm s \vert \beta, \bm z)\right)} \right\}^2 <0,
    \label{eq:AT_GSS}
\end{align}
where $ \tilde W_{D, K, J, J_0, H}(\bm s \vert \beta, \bm z) = \exp \left[ \beta \left( H + J_0 \m + J \sqrt{\q} z^{(1)}\right) s + \beta \left( \dfrac{\beta J^2}{2} (Q-\q) - D\right)s^2\right] $. 
\newline
We note that this result has been previously obtained by the original Authors, via the standard approach \cite{de1978stability}, as reported in \cite{katayama1999ghatak}.
\end{corollary}

\begin{corollary}
For the disordered Blume-Emery-Griffiths-Capel model in the mean field limit deepened by Crisanti and Leuzzi (i.e. by selecting $h=0$, $J_0=0$, $J^2=1$ and $S=1$ in the general Hamiltonian \eqref{eq:Hamiltonian} so to recover the model \eqref{eq:Hamiltonian_BEGC}), the expression of the AT-line reads as 
\begin{align}
    1-\beta^2 \mathbb{E}_1 \left\{ 1- {\beta^2 J^2} \mathbb{E}_1 \left\{ \dfrac{2 e^\xi (2e^\xi + \cosh g(\beta, \bm z^{(1)}))}{(1+ 2e^\xi \cosh g(\beta, \bm z^{(1)}))^2}\right\}^2\right\}<0,
\end{align}
where $\xi=\dfrac{\beta^2}{2} (Q-\q) - \beta D - 4 \beta K Q$ and $g(\beta, \bm z^{(1)})= \beta\sqrt{\q}z^{(1)}$. 
\newline
We note that this expression, to our knowledge, has not been documented in the Literature before.
\end{corollary}

\section{Conclusions and outlooks}\label{Conclusioni}

In this work, we have systematically applied Guerra’s interpolation method \cite{guerra2006replica, guerra_broken} to rigorously analyze the inverse melting and inverse freezing phenomena \cite{SchupperShnerb} in a large class of glassy models. We have focused on several mean-field versions, demonstrating that they all arise as special cases of a broader Hamiltonian. Our approach allowed us to recover known results previously obtained via heuristic techniques such as the replica trick \cite{MPV}, confirming their validity through an independent and mathematically transparent framework. 
\newline
Our findings confirm the robustness of the RS solution in certain regimes of the control parameters and provide explicit analytical expressions for the onset of RSB, captured through  the AT instability line \cite{albanese2023almeida,toninelli2002almeida} adapted to the case. Yet, as instability of the RS solution is present in these models, we investigated also the 1-RSB, providing the explicit expression of the quenched free energy that confirms the heuristic predictions of the replica trick also in this regime.
\newline
Despite these advancements, several open questions remain. First, an entire characterization of the RSB phase beyond the first-step is still to be done: while our analysis successfully recovers the 1-RSB scenario, further extensions towards  a full-RSB solution would provide a more complete picture of the thermodynamic behaviour of these systems. Second, the applicability of our approach to finite-dimensional models, beyond the mean-field approximation, remains to be addressed. As inverse melting and inverse freezing have been experimentally observed in real physical systems, their mean field descriptions are unavoidable simplifications, hence it would be valuable to assess whether our theoretical framework can be extended to capture finite-dimensional effects.
\newline
Future research directions should also focus on the dynamical properties of inverse transitions that have been totally neglected in this manuscript: understanding the kinetics of inverse melting and freezing (such as nucleation processes and the role of metastable states \cite{EmilioEnzo1}) could shed further light on the microscopic mechanisms underlying these anomalous phase transitions. 
\newline
Last but not least, especially nowadays when Machine Learning is spreading across every branch of Physics (see e.g. \cite{Carleo,BarraLenka}), connections with modern neural networks, particularly deep neural networks with energy-based formulations \cite{HopKro1,Deep,DiegoDeep,Albanese2021}, may offer intriguing cross-disciplinary insights, as some disordered neural networks are already known to exhibit behaviour analogous to inverse freezing \cite{BlumeCambelNN1,BlumeCambelNN2,FreezingHopfield}: we plan to report soon on these Blume-Capel neural networks.

\appendix
\section{Proof of Lemma \ref{lemma:1}}
\label{app:lemma1}
The evaluation of the expression of the $t-$derivative of the interpolating quenched statistical pressure has been done via direct computation: 
\begin{align}
    &d_t \mathcal{A}_{N, D, K, J, J_0}(t \vert \beta, \bm z, \bm h) = \dfrac{1}{N} \mathbb{E} \left[ \dfrac{1}{Z_{N, D, K, J, J_0}(t \vert \beta, \bm z, \bm h)}\sum_{\{ \bm s\}} B_{N, D, K, J, J_0}(t \vert \bm s, \bm z, \bm h)\left(  \dfrac{\beta J_0 N }{2} m^2  \right.\right.\notag \\
    &\left.\left.+ \dfrac{\beta J }{4\sqrt{t} \sqrt{N}} \sum_{i<j =1}^N z_{ij} s_i s_j- \beta  D \sum_{i=1}^N s_i^2 - \dfrac{2\beta K }{N} \sum_{i,j} s_i^2 s_j^2 - \psi N m - \dfrac{A}{2  \sqrt{1-t}} \sum_{i=1}^N z_i s_i - \dfrac{B}{2} \sum_{i=1}^N s_i^2\right)\right.,
    \label{eq:dt1app}
\end{align}
where with $B_{N, D, K, J, J_0}(t \vert \bm s, \bm z, \bm h)$ we mean the interpolating Boltzmann factor, that reads as 
\begin{align}
   B_{N, D, K, J, J_0}(t \vert \bm s, \bm z, \bm h)=& \exp \left[ \dfrac{\beta J_0 N t}{2} m^2 + \dfrac{\beta J \sqrt{t}}{2 \sqrt{N}} \sum_{i<j =1}^N z_{ij} s_i s_j - \beta t D \sum_{i=1}^N s_i^2 + \beta \sum_{i=1}^N h_i s_i  \right.\notag \\
    &\left.+ (1-t) \psi N m- \dfrac{2\beta t K }{N} \sum_{i,j} s_i^2 s_j^2+ A \sqrt{1-t} \sum_{i=1}^N z_i s_i + \dfrac{B}{2}(1-t) \sum_{i=1}^N s_i^2 \right].
\end{align}
Applying the Stein's Lemma which states that for a standard Gaussian variable $J$, i.e. $J \sim N(0, 1)$, and for a generic
function $f(J)$ for which the two expectations $\mathbb{E}\left( J f(J)\right)$ and $\mathbb{E}\left( \partial_J f(J)\right)$ both exist, 
\begin{align}
\mathbb{E}_J \left( J f(J)\right)= \mathbb{E}_J \left( \frac{\partial f(J)}{\partial J}\right),
\label{eq:steins_app}
\end{align}
we can write \eqref{eq:dt1app} as
\begin{align}
    d_t \mathcal{A}_{N, D, K, J, J_0}(t \vert \beta, \bm z, \bm h) =& \dfrac{\beta J_0}{2} \l m^2 \r - \beta D \l q_{11} \r - \psi \l m \r   - \dfrac{B}{2} \l q_{11} \r - 2\beta K \l q_{11}^2 \r \notag \\
    &+ \dfrac{\beta J }{4 N \sqrt{t N}} \sum_{i<j=1}^N \mathbb{E} \partial_{z_{ij}} \omega(s_i s_j) - \dfrac{A}{2N \sqrt{1-t}} \sum_{i=1}^N \mathbb{E} \partial_{z_i}\omega(s_i).
    \label{eq:dt2app}
 \end{align}
 Now we need to compute the remaining derivatives. Since they are similar, we compute only the one with respect to $z_{ij}$.

\begin{align}
    \partial_{z_{ij}} \omega(s_i s_j) =& \dfrac{1}{Z_N(t \vert \beta, \bm z, \bm h, J, J_0)} \sum_{\{\bm s\}} \partial_{z_{ij}}(B_{N, D, K, J, J_0}(t \vert \bm s, \bm z, \bm h)) s_i s_j \notag \\
    &+ \sum_{\{\bm s\}} B_{N, D, K, J, J_0}(t \vert \bm s, \bm z, \bm h)s_i s_j \partial_{z_{ij}}\left( \dfrac{1}{Z_N(t \vert \beta, \bm z, \bm h, J, J_0)}\right)\notag \\
    =&\dfrac{\beta J \sqrt{t}}{\sqrt{N}}\left[\left(\dfrac{1}{Z_N(t \vert \beta, \bm z, \bm h, J, J_0)} \sum_{\{\bm s\}} B_{N, D, K, J, J_0}(t \vert \bm s, \bm z, \bm h) (s_i s_j)^2\right) \right. \notag \\ 
    &\left. -\left( \dfrac{1}{Z_N(t \vert \beta, \bm z, \bm h, J, J_0)} \sum_{\{\bm s\}} B_{N, D, K, J, J_0}(t \vert \bm s, \bm z, \bm h) s_i s_j\right)^2\right] \notag \\
    =& \dfrac{\beta J \sqrt{t}}{\sqrt{N}}\left[\omega(s_i s_j)^2 - \omega^2 (s_i s_j)\right].
\end{align}

In this way \eqref{eq:dt2app} becomes 
\begin{align}
    d_t \mathcal{A}_{N, D, K, J, J_0}(t \vert  \beta, \bm z, \bm h) =& \dfrac{\beta J_0}{2} \l m^2 \r - \beta D \l q_{11} \r -2\beta K \l q_{11}^2 \r- \psi \l m \r   - \dfrac{B}{2} \l q_{11} \r \notag \\
    &+ \dfrac{\beta^2 J^2}{4} \left( \l q_{11}^2 \r - \l q_{12}^2 \r\right) - \dfrac{A^2}{2} \left( \l q_{11} \r - \l q_{12} \r \right),
\end{align}
which is the expression of the derivative with respect to $t$ of the interpolating quenched statistical pressure in \eqref{eq:dert}.

\section{Proof of Lemma \ref{lemma:3}}
\label{app:lemma3}
In order to recover the expression of the $t-$derivative of the quenched statistical pressure we have to exploit direct computations and the structure of the interpolating partition function: 
\begin{align}
    d_t& \mathcal{A}_{N, D, K, J, J_0}(t \vert \beta, \bm z, \bm h)= \dfrac{1}{N} \mathbb{E}_0 d_t \left[ \log \exp \left( \mathbb{E}_1 \log \left( \mathbb{E}_2 Z_2^\theta(t \vert \beta, \bm z, \bm h)\right)^{1/\theta}\right)\right] \notag \\
    =&\dfrac{1}{\theta N} \mathbb{E}_0 \mathbb{E}_1 \left[\dfrac{\mathbb{E}_2 \theta Z_2^{\theta-1}(t \vert \beta, \bm z, \bm h) d_t\left( Z_2(t \vert \beta, \bm z, \bm h) \right)}{\mathbb{E}_2 Z_2^\theta(t \vert \beta, \bm z, \bm h)}\right] \notag \\
    =&\dfrac{1}{N} \mathbb{E}_0 \mathbb{E}_1 \left[ \dfrac{1}{\mathbb{E}_2 Z_2^\theta(t \vert \beta, \bm z, \bm h)}\mathbb{E}_2 Z_2^{\theta-1}(t \vert \beta, \bm z, \bm h) \sum_{\{\bm s\}} B_{N, D, K, J, J_0}(t \vert \bm s, \bm z, \bm h) \left(  \dfrac{\beta J}{2\sqrt{2Nt}}\sum_{i,j=1}^Nz_{ij} s_i s_j \right.\right.\notag \\
    &\left.\left. +\dfrac{\beta J_0 N }{2} m^2 - \beta D \sum_{i=1}^N s_i^2 -\psi N m- \sum_{a=1}^2 \dfrac{A^{(a)}}{2\sqrt{1-t}}\sum_{i=1}^N z_i^{(a)} s_i - \dfrac{B}{2} \sum_{i=1}^N s_i^2 + \dfrac{\beta K}{2N} \sum_{i,j=1}^N s_i^2 s_j^2\right) \right].
\end{align}
where we have defined $B_{N, D, K, J, J_0}(t \vert \bm s, \bm z, \bm h)$ as the Boltzmann factor
\begin{align}
    B_{N, D, K, J, J_0}(t \vert \bm s, \bm z, \bm h) =& \exp \left[ \dfrac{\beta J_0 N t }{2} m^2 + \dfrac{\beta J \sqrt{t}}{ \sqrt{2N}} \sum_{i,j =1}^N z_{ij} s_i s_j - \beta t D \sum_{i=1}^N s_i^2 + \beta \sum_{i=1}^N h_i s_i + (1-t) \psi N m \right.\notag \\
    &\left.+ \sum_{a=1}^2 A^{(a)} \sqrt{1-t} \sum_{i=1}^N z_i^{(a)} s_i + \dfrac{B}{2}(1-t) \sum_{i=1}^N s_i^2 +\dfrac{\beta K t}{2N} \sum_{i,j=1}^N s_i^2 s_j^2\right].
\end{align}

Let us first tackle the trivial terms. Using the definition of quenched average and of the order parameters of the model we have 
\begin{align}
\label{eq:dtapp}
    d_t& \mathcal{A}_{N, D, K, J, J_0}(t \vert \beta, \bm z, \bm h)= \dfrac{\beta J_0}{2} \l m^2 \r -\psi \l m\r - \dfrac{B}{2} \l q_{11} \r - \beta D \l q_{11} \r + \dfrac{\beta K }{2} \l q_{11}^2 \r \notag \\
    &+\dfrac{1}{N} \mathbb{E}_0 \mathbb{E}_1 \left[ \dfrac{1}{\mathbb{E}_2 Z_2^\theta(t \vert \beta, \bm z, \bm h)}\mathbb{E}_2\left( Z_2^{\theta-1}(t \vert \beta, \bm z, \bm h) \sum_{\{\bm s\}} B_{N, D, K, J, J_0}(t \vert \bm s, \bm z, \bm h) \dfrac{\beta J }{2\sqrt{2Nt}}\sum_{i,j=1}^N z_{ij} s_i s_j \right)\right] \notag \\
     &-\dfrac{1}{N} \mathbb{E}_0 \mathbb{E}_1 \left[ \dfrac{1}{\mathbb{E}_2 Z_2^\theta(t \vert \beta, \bm z, \bm h)}\mathbb{E}_2\left( Z_2^{\theta-1}(t \vert \beta, \bm z, \bm h)\sum_{\{\bm s\}} B_{N, D, K, J, J_0}(t \vert \bm s, \bm z, \bm h)\sum_{a=1}^2 \dfrac{A^{(a)}}{2\sqrt{1-t}}\sum_{i=1}^N z_i^{(a)} s_i \right)\right].
\end{align}
Now we analyze the last two terms separately. Let us start with the former.

\begin{align}
    \dfrac{1}{N} \mathbb{E}_0 \mathbb{E}_1 &\left[ \dfrac{1}{\mathbb{E}_2 Z_2^\theta(t \vert \beta, \bm z, \bm h)}\mathbb{E}_2\left( Z_2^{\theta-1}(t \vert \beta, \bm z, \bm h)\sum_{\{\bm s\}} B_{N, D, K, J, J_0}(t \vert \bm s, \bm z, \bm h) \dfrac{\beta J }{2\sqrt{2Nt}}\sum_{i,j=1}^N z_{ij} s_i s_j \right)\right] \notag \\
    &= \dfrac{\beta J }{2N\sqrt{2Nt}} \sum_{i,j=1}^N \mathbb{E}_0 \mathbb{E}_1 \left[ \partial_{z_{ij}} \left( \dfrac{\mathbb{E}_2\left( Z_2^{\theta-1}(t \vert \beta, \bm z, \bm h)\sum_{\{\bm s\}} B_{N, D, K, J, J_0}(t \vert \bm s, \bm z, \bm h) s_i s_j \right)}{\mathbb{E}_2 Z_2^\theta(t \vert \beta, \bm z, \bm h)}\right)\right] \notag \\
    &= \dfrac{\beta J }{2N\sqrt{2Nt}} \sum_{i,j=1}^N \mathbb{E}_0 \mathbb{E}_1 \left[ \dfrac{\partial_{z_{ij}} \left( \mathbb{E}_2\left( Z_2^{\theta-1}(t \vert \beta, \bm z, \bm h)\sum_{\{\bm s\}} B_{N, D, K, J, J_0}(t \vert \bm s, \bm z, \bm h) s_i s_j \right) \right)}{\mathbb{E}_2 Z_2^\theta(t \vert \beta, \bm z, \bm h)} \right.\notag \\
    &\left.- \dfrac{\mathbb{E}_2\left( Z_2^{\theta-1}(t \vert \beta, \bm z, \bm h)\sum_{\{\bm s\}} B_{N, D, K, J, J_0}(t \vert \bm s, \bm z, \bm h) s_i s_j \right)}{\left(\mathbb{E}_2 Z_2^\theta(t \vert \beta, \bm z, \bm h)\right)^2} \mathbb{E}_2 \partial_{z_{ij}} \left( Z_2^\theta(t \vert \beta, \bm z, \bm h\right)\right],
    \label{eq:dzij}
\end{align}
where we have applied Stein's Lemma \eqref{eq:steins_app}.
Therefore, we need to compute the two derivative in \eqref{eq:dzij} apart: 
\begin{align}
    &\partial_{z_{ij}} \left[\mathbb{E}_2\left( Z_2^{\theta-1}(t \vert \beta, \bm z, \bm h)\sum_{\{\bm s\}} B(\bm s \vert \bm z, \bm z, \bm h) s_i s_j \right) \right]= \notag \\
    &\hspace{3cm}\mathbb{E}_2  \left[ (\theta-1) Z_2^{\theta-2} (t \vert \beta, \bm z, \bm h) \dfrac{\beta J \sqrt{t}}{\sqrt{2N}} \left(\sum_{\{\bm s\}} B_{N, D, K, J, J_0}(t \vert \bm s, \bm z, \bm h)s_i s_j\right)^2 \right. \notag \\
    &\hspace{3cm}\left.+ \dfrac{\beta J \sqrt{t}}{\sqrt{2N}}Z_2^{\theta-1}(t \vert \beta, \bm z, \bm h) \sum_{\{\bm s\}}B_{N, D, K, J, J_0}(t \vert \bm s, \bm z, \bm h)(s_i s_j)^2\right] \label{eq:dzij1}\\
    &\partial_{z_{ij}} \left( Z_2^\theta(t \vert \beta, \bm z, \bm h)\right) = \dfrac{\beta J \sqrt{t}}{\sqrt{2N}}\theta Z_2^{\theta-1}(t \vert \beta, \bm z, \bm h)  \sum_{\{\bm s\}}B_{N, D, K, J, J_0}(t \vert \bm s, \bm z, \bm h) s_i s_j. \label{eq:dzij2}
\end{align}
Putting \eqref{eq:dzij1}-\eqref{eq:dzij2} in \eqref{eq:dzij} we get 
\begin{align}
    &\dfrac{\beta^2 J^2 }{4N^2} \sum_{ij} \mathbb{E}_0 \mathbb{E}_1 \left[ \dfrac{\mathbb{E}_2 \left( Z_2^{\theta-1}(t \vert \beta, \bm z, \bm h) \sum_{\{\bm s\}} B_{N, D, K, J, J_0}(t \vert \bm s, \bm z, \bm h) (s_is_j)^2\right)}{\mathbb{E}_2 Z_2^\theta(t \vert \beta, \bm z, \bm h)} \right. \notag \\
    &\hspace{2cm}\left.+ (\theta-1)\dfrac{\mathbb{E}_2 \left( Z_2^{\theta-2}(t \vert \beta, \bm z, \bm h) \left(\sum_{\{\bm s\}} B_{N, D, K, J, J_0}(t \vert \bm s, \bm z, \bm h) (s_is_j)\right)^2\right)}{\mathbb{E}_2 Z_2^\theta(t \vert \beta, \bm z, \bm h)}  \right.\notag \\
    &\left.\hspace{2cm}- \theta \left(\dfrac{\mathbb{E}_2\left( Z_2^{\theta-1}(t \vert \beta, \bm z, \bm h)\sum_{\{\bm s\}} B_{N, D, K, J, J_0}(t \vert \bm s, \bm z, \bm h) s_i s_j \right)}{\left(\mathbb{E}_2 Z_2^\theta(t \vert \beta, \bm z, \bm h)\right)} \right)^2\right], 
\end{align}
which can be rewritten, exploiting the definition of the quenched average and of the order parameters as 
\begin{align}
    \dfrac{\beta^2 J^2 }{4} \left[ \l q_{11}^2 \r + (\theta-1) \l q_{12} \r_2 - \theta \l q_{12} \r_1 \right].
\end{align}

Now we pass to the second term of \eqref{eq:dtapp}:
\begin{align}
    -\dfrac{1}{N} \mathbb{E}_0 \mathbb{E}_1& \left[ \dfrac{1}{\mathbb{E}_2 Z_2^\theta(t \vert \beta, \bm z, \bm h)}\mathbb{E}_2\left( Z_2^{\theta-1}(t \vert \beta, \bm z, \bm h)\right.\right.\notag \\
    &\left.\left.\cdot\sum_{\{\bm s\}} B_{N, D, K, J, J_0}(t \vert \bm s, \bm z, \bm h)\left(\dfrac{A^{(1)}}{2\sqrt{1-t}}\sum_{i=1}^N z_i^{(1)} s_i + \dfrac{A^{(2)}}{2\sqrt{1-t}}\sum_{i=1}^N z_i^{(2)} s_i\right)\right)\right]. 
\end{align}
We focus on the two terms individually. Applying Stein's Lemma \eqref{eq:steins_app}
\begin{align}
      &-\dfrac{A^{(1)}}{2N\sqrt{1-t}}\sum_{i=1}^N\mathbb{E}_0 \mathbb{E}_1 \left[ \dfrac{1}{\mathbb{E}_2 Z_2^\theta(t \vert \beta, \bm z, \bm h)}\mathbb{E}_2\left( Z_2^{\theta-1}(t \vert \beta, \bm z, \bm h)\sum_{\{\bm s\}} B_{N, D, K, J, J_0}(t \vert \bm s, \bm z, \bm h) z_i^{(1)} s_i \right)\right] \notag \\
      &=-\dfrac{A^{(1)}}{2N\sqrt{1-t}}\sum_{i=1}^N\mathbb{E}_0 \mathbb{E}_1 \left[ \partial_{z_i^{(1)}} \left(\dfrac{1}{\mathbb{E}_2 Z_2^\theta(t \vert \beta, \bm z, \bm h)}\mathbb{E}_2\left( Z_2^{\theta-1}(t \vert \beta, \bm z, \bm h)\sum_{\{\bm s\}} B_{N, D, K, J, J_0}(t \vert \bm s, \bm z, \bm h) s_i \right)\right)\right] \notag \\
      &= -\dfrac{A^{(1)}}{2N\sqrt{1-t}}\sum_{i=1}^N\mathbb{E}_0 \mathbb{E}_1 \left[\dfrac{ \mathbb{E}_2\partial_{z_i^{(1)}}\left( Z_2^{\theta-1}(t \vert \beta, \bm z, \bm h)\sum_{\{\bm s\}} B_{N, D, K, J, J_0}(t \vert \bm s, \bm z, \bm h) s_i \right)}{\mathbb{E}_2 Z_2^\theta(t \vert \beta, \bm z, \bm h)} \right. \notag \\
      &\left.\hspace{1cm}- \dfrac{\mathbb{E}_2\left( Z_2^{\theta-1}(t \vert \beta, \bm z, \bm h)\sum_{\{\bm s\}} B_{N, D, K, J, J_0}(t \vert \bm s, \bm z, \bm h) s_i \right)}{\left(\mathbb{E}_2 Z_2^\theta(t \vert \beta, \bm z, \bm h)\right)^2} \mathbb{E}_2\partial_{z_i^{(1)}} Z_2^{\theta}(t \vert \beta, \bm z, \bm h)\right].
      \label{eq:dz1}
\end{align}
As we have already done for the derivative w.r.t. $z_{ij}$, we can compute the two terms in \eqref{eq:dz1} in order to have 
\begin{align}
    \partial_{z_i^{(1)}}&\left( Z_2^{\theta-1}(t \vert \beta, \bm z, \bm h)\sum_{\{\bm s\}} B_{N, D, K, J, J_0}(t \vert \bm s, \bm z, \bm h) s_i \right) \notag \\
    &= A^{(1)}\sqrt{1-t}\left[(\theta-1) Z_2^{\theta-2}(t \vert \beta, \bm z, \bm h) \left( \sum_{\{\bm s\}} B_{N, D, K, J, J_0}(t \vert \bm s, \bm z, \bm h) s_i\right)^2 \right.\notag \\
    &\left.+ Z_2^{\theta-1}(t \vert \beta, \bm z, \bm h)\sum_{\{\bm s\}} B_{N, D, K, J, J_0}(t \vert \bm s, \bm z, \bm h) s_i^2\right] \\
    \partial_{z_i^{(1)}} &Z_2^{\theta}(t \vert \beta, \bm z, \bm h) = A^{(1)} \sqrt{1-t} \theta Z_2^{\theta-1}(t \vert \beta, \bm z, \bm h) \sum_{\{\bm s\}} B_{N, D, K, J, J_0}(t \vert \bm s, \bm z, \bm h) s_i.
\end{align}
Rearranging all the terms together, exploiting the definition of the quenched average and of the order parameters, we can rewrite \eqref{eq:dz1} as 
\begin{align}
    -\dfrac{(A^{(1)})^2}{2} \left( \l q_{11} \r +(\theta-1) \l q_{12} \r_2 - \theta \l q_{12} \r_1 \right).
    \label{eq:dz1app}
\end{align}

As far as the second term under investigation in \eqref{eq:dz1} concerns, we have
\begin{align}
     &-\dfrac{A^{(2)}}{2N\sqrt{1-t}}\sum_{i=1}^N\mathbb{E}_0 \mathbb{E}_1 \left[ \dfrac{1}{\mathbb{E}_2 Z_2^\theta(t \vert \beta, \bm z, \bm h)}\mathbb{E}_2\left( Z_2^{\theta-1}(t \vert \beta, \bm z, \bm h)\sum_{\{\bm s\}} B_{N, D, K, J, J_0}(t \vert \bm s, \bm z, \bm h) z_i^{(2)} s_i \right)\right] \notag \\
      &=-\dfrac{A^{(2)}}{2N\sqrt{1-t}}\sum_{i=1}^N\mathbb{E}_0 \mathbb{E}_1 \left[ \partial_{z_i^{(2)}} \left(\dfrac{1}{\mathbb{E}_2 Z_2^\theta(t \vert \beta, \bm z, \bm h)}\mathbb{E}_2\left( Z_2^{\theta-1}(t \vert \beta, \bm z, \bm h)\sum_{\{\bm s\}} B_{N, D, K, J, J_0}(t \vert \bm s, \bm z, \bm h) s_i \right)\right)\right] \notag \\
      &= -\dfrac{A^{(2)}}{2N\sqrt{1-t}}\sum_{i=1}^N\mathbb{E}_0 \mathbb{E}_1 \left[\dfrac{ \mathbb{E}_2\partial_{z_i^{(2)}}\left( Z_2^{\theta-1}(t \vert \beta)\sum_{\{\bm s\}} B_{N, D, K, J, J_0}(t \vert \bm s, \bm z, \bm h) s_i \right)}{\mathbb{E}_2 Z_2^\theta(t \vert \beta, \bm z, \bm h)} \right] \notag \\
      &= -\dfrac{(A^{(2)})^2}{2} \sum_{i=1}^N\mathbb{E}_0 \mathbb{E}_1 \left[\dfrac{ \mathbb{E}_2(\theta-1) Z_2^{\theta-2}(t \vert \beta, \bm z, \bm h) \left( \sum_{\{\bm s\}}B_{N, D, K, J, J_0}(t \vert \bm s, \bm z, \bm h) s_i\right)^2 }{\mathbb{E}_2 Z_2^\theta(t \vert \beta, \bm z, \bm h)}\right. \notag \\
      &\hspace{1cm}\left.+\dfrac{  Z_2^{\theta-1}(t \vert \beta,\bm z, \bm h)\sum_{\{\bm s\}} B_{N, D, K, J, J_0}(t \vert \bm s, \bm z, \bm h) s_i^2}{\mathbb{E}_2 Z_2^\theta(t \vert \beta, \bm z, \bm h)} \right],
      \label{eq:dz2}
\end{align}
which can be rewritten as 
\begin{align}
    -\dfrac{(A^{(2)})^2}{2} \left( \l q_{11} \r + (\theta-1) \l q_{12} \r_2 \right).
    \label{eq:dz2app}
\end{align}
We stress that the term $\mathbb{E}_2 Z_2^\theta(t \vert \beta)$ is independent on $z^{(2)}$. This is the reason why there is a missing term with respect to the $z^{(1)}-$derivative. 
Putting together \eqref{eq:dtapp} with \eqref{eq:dz1app} and \eqref{eq:dz2app} we reach the thesis.

\section{Contributions of the sub-leading terms}
\label{app:ABC}

\begin{align}
    \mathcal{A}(\m, Q, \qb_1, \qb_2)=& \dfrac{\log \sum_{\{\bm s\}} \tilde W_{D, K, J, J_0}(\bm s \vert \beta, \bm z, \bm h)}{\sum_{\{\bm s\}} \tilde W_{D, K, J, J_0}(\bm s \vert \beta, \bm z, \bm h)} \notag \\
    &- \dfrac{\sum_{\{\bm s\}} \tilde W_{D, K, J, J_0}(\bm s \vert \beta, \bm z, \bm h) s \mathbb{E}_2 \sum_{\{\bm s\}}W_{D, K, J, J_0}(\bm s \vert \beta, \bm z, \bm h)\log \sum_{\bm s } W_{D, K, J, J_0}(\bm s \vert \beta, \bm z, \bm h)}{\left( \sum_{\{\bm s\}} \tilde W_{D,K, J, J_0}(\bm s \vert \beta, \bm z, \bm h)\right)^2} \\
    \mathcal{B}(\m, Q, \qb_1, \qb_2)=& \dfrac{\mathbb{E}_2 \log \sum_{\{\bm s\}} W_{D, K, J, J_0}(\bm s \vert \beta, \bm z, \bm h)}{\sum_{\{\bm s\}} \tilde W_{D, K, J, J_0}(\bm s \vert \beta, \bm z, \bm h)} \notag \\
    &- \dfrac{\sum_{\{\bm s\}} \tilde W_{D, K, J, J_0}(\bm s \vert \beta, \bm z, \bm h)s^2 \mathbb{E}_2 \sum_{\{\bm s\}} W_{D, K, J, J_0}(\bm s \vert \beta, \bm z, \bm h)\log \sum_{\bm s } W_{D, K, J, J_0}(\bm s \vert \beta, \bm z, \bm h)}{\left( \sum_{\{\bm s\}} \tilde W_{D, K, J, J_0}(\bm s \vert \beta, \bm z, \bm h)\right)^2} \\
    \mathcal{C}(\m, Q, \qb_1, \qb_2)=& 2 \left[ \dfrac{\sum_{\{\bm s\}} \tilde W_{D, K, J, J_0}(\bm s \vert \beta, \bm z, \bm h) s }{\sum_{\{\bm s\}} \tilde W_{D, K, J, J_0}(\bm s \vert \beta, \bm z, \bm h)} \left( \dfrac{\mathbb{E}_2 \log \sum_{\{\bm s\}} W_{D, K, J, J_0}(\bm s \vert \beta, \bm z, \bm h)}{\sum_{\{\bm s\}} \tilde W_{D, K, J, J_0}(\bm s \vert \beta, \bm z, \bm h)} \right.\right.\notag \\
    &\left.\left.- \dfrac{\sum_{\{\bm s\}} \tilde W_{D, K, J, J_0}(\bm s \vert \beta, \bm z, \bm h) s \mathbb{E}_2 \sum_{\bm s } W_{D, K, J, J_0}(\bm s \vert \beta, \bm z, \bm h) \log \sum_{\{\bm s\}} W_{D, K, J, J_0}(\bm s \vert \beta, \bm z, \bm h)}{\left( \sum_{\{\bm s\}} \tilde W_{D, K, J, J_0}(\bm s \vert \beta, \bm z, \bm h)\right)^2}\right) \right] \\
   \mathcal E(\m, Q, \qb_1, \qb_2) =& 2 \dfrac{\mathbb{E}_2 \left( \sum_{\{\bm s\}} W_{D, K, J, J_0}(\bm s \vert \beta, \bm z, \bm h) s \right)^2 \log \sum_{\{\bm s\}} W_{D, K, J, J_0}(\bm s \vert \beta, \bm z, \bm h)}{\sum_{\{\bm s\}} \tilde W_{D, K, J, J_0}(\bm s \vert \beta, \bm z, \bm h)} \notag \\
   &- \dfrac{\mathbb{E}_2 \left( \sum_{\{\bm s\}} W_{D, K, J, J_0}(\bm s \vert \beta, \bm z, \bm h) s\right)^2\log \sum_{\{\bm s\}} W_{D, K, J, J_0}(\bm s \vert \beta, J)}{\left( \sum_{\{\bm s\}} \tilde W_{D, K, J, J_0}(\bm s \vert \beta, \bm z, \bm h) \right)^2},
\end{align}
where we recall that 
\begin{align}
    W_{D, K, J, J_0}(\bm s \vert \beta, \bm z, \bm h)=& \exp \left(\beta s \left( h + J_0 \m + J \sqrt{\q_1} z^{(1)} + J \sqrt{\q_2-\q_1} z^{(2)} \right) + \beta s^2\left( \dfrac{\beta J^2}{2}(Q-\q_2) - D \right) \right), \\
    \tilde W_{D, K, J, J_0}(\bm s \vert \beta, \bm z, \bm h)=& \exp \left(\beta s \left( h + J_0 \m + J \sqrt{\q} z \right) + \beta s^2\left( \dfrac{\beta J^2}{2}(Q-\q) - D \right) \right).
\end{align}

\section{Computation of \eqref{eq:K}}
\label{app:K}
In order to compute $K(\m, \q, \q_2)$, we rewrite $\mathcal{A}_{1-RSB}(\theta \vert \m, Q, \q, \q_2)$ exploiting the expansions of the order parameters: 
\begin{align}
    \mathcal{A}_{1-RSB}(\theta \vert \m, Q, \q, \q_2) =& \dfrac{1}{\theta} \mathbb{E}_1 \log \mathbb{E}_2 \left\{ \sum_{\{\bm s\}} \mathcal{W}(\bm s \vert \theta, \beta, J)\right\}^\theta - \dfrac{\beta J_0}{2} \left[ \m^2 +2\m A(\theta-1)\right] \notag \\
    &- \dfrac{\beta^2 J^2}{4} \left[ Q^2 + 2B(\theta-1)Q + (\theta-1) \tilde q_2^2 - (\theta-1) \q^2 -\q^2 -2(\theta -1) C\q \right] \notag \\
    &+ \dfrac{\beta K}{2} \left[ Q^2 +2BQ(\theta-1)\right],
\end{align}
where 
\begin{align}
    \mathcal{W}(\bm s \vert \theta, \beta, J)=&\beta s \left( h + J_0 (m+ (\theta-1) A) + J \sqrt{\q + (\theta-1) C} z^{(1)} + J \sqrt{\tilde q_2 -\q + (\theta-1) (E-C)} z^{(2)}\right)\notag \\
    &+ \beta s^2 \left( \dfrac{\beta J^2 }{2} \left( Q - \tilde q_2 +(\theta-1) (B-E) - D - K (Q+(\theta-1) B)\right)\right).
\end{align}
Now we need to derive with respect to $\theta$. 
\begin{align}
     d_\theta \mathcal{A}_{1-RSB}(\theta \vert \m, Q, \q, \q_2) =& -\dfrac{\beta^2 J^2}{2} B Q - \dfrac{\beta^2 J^2}{4} \tilde q_2 + \dfrac{\beta^2 J^2}{4} \q^2 + \dfrac{\beta^2 J^2}{2} C \q - \beta J_0 A + \beta^2 K Q B  \notag \\
     &+d_\theta\left( \dfrac{1}{\theta} \mathbb{E}_1 \log \mathbb{E}_2 \left\{ \sum_{\{\bm s\}} \mathcal{W}(\bm s \vert \theta, \beta, J)\right\}^\theta\right) .
     \label{dtAapp}
\end{align}
We compute the last term individually. 
\begin{align}
    &d_\theta\left( \dfrac{1}{\theta} \mathbb{E}_1 \log \mathbb{E}_2 \left\{ \sum_{\{\bm s\}} \mathcal{W}(\bm s \vert \theta, \beta, J)\right\}^\theta\right)=-\dfrac{1}{\theta^2 }\mathbb{E}_1 \log \mathbb{E}_2 \left\{ \sum_{\{\bm s\}} \mathcal{W}(\bm s \vert \theta, \beta, J)\right\}^\theta \notag \\
    &+ \dfrac{1}{\theta} \mathbb{E}_1 \left\{ \dfrac{1}{\mathbb{E}_2 \left(\sum_{\{\bm s\}} \mathcal{W}(\bm s \vert \theta, \beta, J)\right)^\theta}\mathbb{E}_2 \left[ \left(\sum_{\{\bm s\}} \mathcal{W}(\bm s \vert \theta, \beta, J)\right)^\theta \log \sum_{\{\bm s\}} \mathcal{W}(\bm s \vert \theta, \beta, J) \right]\right\} \notag \\
    &+\mathbb{E}_1 \left\{ \dfrac{1}{\mathbb{E}_2 \left(\sum_{\{\bm s\}} \mathcal{W}(\bm s \vert \theta, \beta, J)\right)^\theta}\mathbb{E}_2 \left[ \left(\sum_{\{\bm s\}} \mathcal{W}(\bm s \vert \theta, \beta, J)\right)^{\theta-1} \sum_{\{\bm s\}} \mathcal{W}(\bm s \vert \theta, \beta, J) \right. \right.\notag \\
    &\left.\left.\cdot\left( \beta s^2 \left( \dfrac{\beta J^2}{2} (B-E) - B K \right) + \beta s \left( A J_0 + \dfrac{C J z^{(1)}}{2 \sqrt{\q + (\theta-1) C} } + \dfrac{(E-C) J z^{(2)}}{2 \sqrt{\tilde q_2 - \q + (\theta-1) (E-C)}}\right)\right)\right]\right\}.
    \label{eq:dtK}
\end{align}
Now we focus on the last term  of \eqref{eq:dtK}. 
\begin{align}
 &\beta \left( \dfrac{\beta J^2}{2} (B-E) - B K \right)\mathbb{E}_1 \left\{ \dfrac{\mathbb{E}_2 \left[ \left(\sum_{\{\bm s\}} \mathcal{W}(\bm s \vert \theta, \beta, J)\right)^{\theta-1} \sum_{\{\bm s\}} \mathcal{W}(\bm s \vert \theta, \beta, J) s^2 \right]}{\mathbb{E}_2 \left(\sum_{\{\bm s\}} \mathcal{W}(\bm s \vert \theta, \beta, J)\right)^\theta} \right\} \notag \\
 &+ \beta \left( A J_0 + \dfrac{C J }{2 \sqrt{\q + (\theta-1) C} } \right)\mathbb{E}_1 \left\{ \dfrac{\mathbb{E}_2 \left[ \left(\sum_{\{\bm s\}} \mathcal{W}(\bm s \vert \theta, \beta, J)\right)^{\theta-1} \sum_{\{\bm s\}} \mathcal{W}(\bm s \vert \theta, \beta, J) z^{(1)} s \right]}{\mathbb{E}_2 \left(\sum_{\{\bm s\}} \mathcal{W}(\bm s \vert \theta, \beta, J)\right)^\theta} \right\}\notag \\
 &+ \beta \left( \dfrac{(E-C) J }{2 \sqrt{\tilde q_2 - \q + (\theta-1) (E-C)}}\right)\mathbb{E}_1 \left\{ \dfrac{\mathbb{E}_2 \left[ \left(\sum_{\{\bm s\}} \mathcal{W}(\bm s \vert \theta, \beta, J)\right)^{\theta-1} \sum_{\{\bm s\}} \mathcal{W}(\bm s \vert \theta, \beta, J) z^{(2)} s \right]}{\mathbb{E}_2 \left(\sum_{\{\bm s\}} \mathcal{W}(\bm s \vert \theta, \beta, J)\right)^\theta} \right\} \notag \\
 &=\beta \left( \dfrac{\beta J^2}{2} (B-E) - B K \right)\mathbb{E}_1 \left\{ \dfrac{\mathbb{E}_2 \left[ \left(\sum_{\{\bm s\}} \mathcal{W}(\bm s \vert \theta, \beta, J)\right)^{\theta-1} \sum_{\{\bm s\}} \mathcal{W}(\bm s \vert \theta, \beta, J) s^2 \right]}{\mathbb{E}_2 \left(\sum_{\{\bm s\}} \mathcal{W}(\bm s \vert \theta, \beta, J)\right)^\theta} \right\} \notag \\
 &+ \beta \left( A J_0 \right)\mathbb{E}_1 \left\{ \dfrac{\mathbb{E}_2 \left[ \left(\sum_{\{\bm s\}} \mathcal{W}(\bm s \vert \theta, \beta, J)\right)^{\theta-1} \sum_{\{\bm s\}} \mathcal{W}(\bm s \vert \theta, \beta, J) s \right]}{\mathbb{E}_2 \left(\sum_{\{\bm s\}} \mathcal{W}(\bm s \vert \theta, \beta, J)\right)^\theta} \right\}\notag \\
 &+ \beta \left( \dfrac{C J}{2 \sqrt{\q + (\theta-1) C} } \right)\mathbb{E}_1 \partial_{ z^{(1)} }\left\{ \dfrac{\mathbb{E}_2 \left[ \left(\sum_{\{\bm s\}} \mathcal{W}(\bm s \vert \theta, \beta, J)\right)^{\theta-1} \sum_{\{\bm s\}} \mathcal{W}(\bm s \vert \theta, \beta, J) s\right]}{\mathbb{E}_2 \left(\sum_{\{\bm s\}} \mathcal{W}(\bm s \vert \theta, \beta, J)\right)^\theta} \right\}\notag \\
 &+ \beta \left( \dfrac{(E-C) J }{2 \sqrt{\tilde q_2 - \q + (\theta-1) (E-C)}}\right)\mathbb{E}_1 \left\{ \dfrac{\mathbb{E}_2 \partial_{ z^{(2)}}\left[ \left(\sum_{\{\bm s\}} \mathcal{W}(\bm s \vert \theta, \beta, J)\right)^{\theta-1} \sum_{\{\bm s\}} \mathcal{W}(\bm s \vert \theta, \beta, J) s \right]}{\mathbb{E}_2 \left(\sum_{\{\bm s\}} \mathcal{W}(\bm s \vert \theta, \beta, J)\right)^\theta} \right\},
\end{align}
where in the last passage we have applied Stein's Lemma \eqref{eq:steins_app}.
This allows us to write the term $d_\theta\left( \dfrac{1}{\theta} \mathbb{E}_1 \log \mathbb{E}_2 \left\{ \sum_{\{\bm s\}} \mathcal{W}(\bm s \vert \theta, \beta, J)\right\}^\theta\right)$ as 
\begin{align}
    &d_\theta\left( \dfrac{1}{\theta} \mathbb{E}_1 \log \mathbb{E}_2 \left\{ \sum_{\{\bm s\}} \mathcal{W}(\bm s \vert \theta, \beta, J)\right\}^\theta\right)=-\dfrac{1}{\theta^2 }\mathbb{E}_1 \log \mathbb{E}_2 \left\{ \sum_{\{\bm s\}} \mathcal{W}(\bm s \vert \theta, \beta, J)\right\}^\theta \notag \\
    &+ \dfrac{1}{\theta} \mathbb{E}_1 \left\{ \dfrac{1}{\mathbb{E}_2 \left(\sum_{\{\bm s\}} \mathcal{W}(\bm s \vert \theta, \beta, J)\right)^\theta}\mathbb{E}_2 \left[ \left(\sum_{\{\bm s\}} \mathcal{W}(\bm s \vert \theta, \beta, J)\right)^\theta \log \sum_{\{\bm s\}} \mathcal{W}(\bm s \vert \theta, \beta, J) \right]\right\} \notag \\ 
    &+\beta \left( \dfrac{\beta J^2}{2} (B-E) - B K \right)\mathbb{E}_1 \left\{ \dfrac{\mathbb{E}_2 \left[ \left(\sum_{\{\bm s\}} \mathcal{W}(\bm s \vert \theta, \beta, J)\right)^{\theta-1} \sum_{\{\bm s\}} \mathcal{W}(\bm s \vert \theta, \beta, J) s^2 \right]}{\mathbb{E}_2 \left(\sum_{\{\bm s\}} \mathcal{W}(\bm s \vert \theta, \beta, J)\right)^\theta} \right\} \notag \\
    &+ \beta \left( A J_0 \right)\mathbb{E}_1 \left\{ \dfrac{\mathbb{E}_2 \left[ \left(\sum_{\{\bm s\}} \mathcal{W}(\bm s \vert \theta, \beta, J)\right)^{\theta-1} \sum_{\{\bm s\}} \mathcal{W}(\bm s \vert \theta, \beta, J) s \right]}{\mathbb{E}_2 \left(\sum_{\{\bm s\}} \mathcal{W}(\bm s \vert \theta, \beta, J)\right)^\theta} \right\}\notag \\
    &- \dfrac{\beta^2 J^2 \theta C}{2} \mathbb{E}_1 \left\{ 
 \dfrac{\mathbb{E}_2 \left[ \left(\sum_{\{\bm s\}} \mathcal{W}(\bm s \vert \theta, \beta, J)\right)^{\theta-1} \sum_{\{\bm s\}} \mathcal{W}(\bm s \vert \theta, \beta, J) s \right]}{\mathbb{E}_2 \left(\sum_{\{\bm s\}} \mathcal{W}(\bm s \vert \theta, \beta, J)\right)^\theta}\right\}^2 \notag \\
 &+ \dfrac{\beta^2 J^2 D}{2} \mathbb{E}_1 \left\{ \dfrac{\mathbb{E}_2 \left[ \left(\sum_{\{\bm s\}} \mathcal{W}(\bm s \vert \theta, \beta, J)\right)^{\theta-1} \sum_{\{\bm s\}} \mathcal{W}(\bm s \vert \theta, \beta, J) s^2 \right]}{\mathbb{E}_2 \left(\sum_{\{\bm s\}} \mathcal{W}(\bm s \vert \theta, \beta, J)\right)^\theta} \right\} \notag \\
  &+ \dfrac{\beta^2 J^2 D}{2} (\theta-1) \mathbb{E}_1 \left\{ 
 \dfrac{\mathbb{E}_2 \left[ \left(\sum_{\{\bm s\}} \mathcal{W}(\bm s \vert \theta, \beta, J)\right)^{\theta-2} \left( \sum_{\{\bm s\}} \mathcal{W}(\bm s \vert \theta, \beta, J) s\right)^2 \right]}{\mathbb{E}_2 \left(\sum_{\{\bm s\}} \mathcal{W}(\bm s \vert \theta, \beta, J)\right)^\theta}\right\}.
\end{align}

Inserting this in \eqref{dtAapp}, we find the first order expression of the derivative of the 1-RSB quenched statistical pressure with respect to $\theta$. If we suppose $\theta=1$ we find \eqref{eq:K}.

\acknowledgments

The Authors are grateful to Luca Leuzzi, Alberto Fachechi and Andrea Alessandrelli for precious comments on this research.
\newline
LA and AB acknowledge the PRIN 2022 grant {\em Statistical Mechanics of Learning Machines} number 20229T9EAT funded by the Italian Ministry of University and Research (MUR) in the framework of European Union - Next Generation EU.\\
LA acknowledges funding also by the PRIN 2022 grant {\em “Stochastic Modeling of Compound Events (SLIDE)”} n. P2022KZJTZ funded by the Italian Ministry of University and Research (MUR) in the framework of European Union - Next Generation EU.
\newline
AB acknowledges further support by Sapienza Università di Roma (via the grant {\em Statistical learning theory for generalized Hopfield models}), prot. n. RM12419112BF7119.
\newline
ENMC thanks the PRIN 2022 project ``Mathematical Modelling of Heterogeneous Systems (MMHS)",
financed by the European Union - Next Generation EU,
CUP B53D23009360006, Project Code 2022MKB7MM, PNRR M4.C2.1.1.
\newline
All the authors are members of the  GNFM group within INdAM which is acknowledged too.

\end{document}